\documentclass[conference]{IEEEtran}
\IEEEoverridecommandlockouts
\usepackage{cite}
\usepackage{amsmath,amssymb,amsfonts,amsthm}
\usepackage{algorithmic}
\usepackage{graphicx}
\usepackage{textcomp}
\usepackage{xcolor}
\usepackage{xspace}
\usepackage{lineno}
\usepackage{booktabs}

\usepackage{graphicx}
\usepackage{subcaption}
\usepackage[linesnumbered, ruled, vlined]{algorithm2e}
\usepackage[colorlinks=true, allcolors=blue]{hyperref}

\usepackage{tikz}       
\usepackage{array}      
\usepackage{booktabs}    
\usepackage{cleveref}

\newtheorem{definition}{Definition}

\newtheorem{problem}{Problem}
\newtheorem{lemma}{Lemma}
\def\BibTeX{{\rm B\kern-.05em{\sc i\kern-.025em b}\kern-.08em
    T\kern-.1667em\lower.7ex\hbox{E}\kern-.125emX}}

\begin{document}
\newcommand{\shengxin}[1]{{\color{orange}#1}\xspace}
\newcommand{\yinyu}[1]{{\color{cyan}#1}\xspace}
\newcommand{\Yui}{\color{blue}}

\title{Temporal $k$-Core Query, Revisited}

\author{\IEEEauthorblockN{Yinyu Liu$^{\hspace{.1em}\dag}$, Kaiqiang Yu$^{\hspace{.1em}\ddag}$, Shengxin Liu$^{\hspace{.1em}\dag}$, Cheng Long$^{\hspace{.1em}\ddag}$, Zhaoquan Gu$^{\hspace{.1em}\dag}$}
 \vspace{.4em}
 \IEEEauthorblockA{\textit{$^\dag$Harbin Institute of Technology, Shenzhen, China} \\
  \textit{$^\ddag$Nanyang Technological University, Singapore} \\ 
  \vspace{4pt}
  \{23s151053@stu., sxliu@, guzhaoquan@\}hit.edu.cn, \{kaiqiang002@e., c.long@\}ntu.edu.sg}
}

\maketitle

\begin{abstract}

Querying cohesive subgraphs in temporal graphs is essential for understanding the dynamic structure of real-world networks, such as evolving communities in social platforms, shifting hyperlink structures on the Web, and transient communication patterns in call networks. Recently, research has focused on the temporal $k$-core query, which aims to identify all $k$-cores across all possible time sub-intervals within a given query interval.
The state-of-the-art algorithm \texttt{OTCD} mitigates redundant computations over overlapping sub-intervals by exploiting inclusion relationships among $k$-cores in different time intervals. Nevertheless, \texttt{OTCD} remains limited in scalability due to the combinatorial growth in interval enumeration and repeated processing. In this paper, we revisit the temporal $k$-core query problem and introduce a novel algorithm \texttt{CoreT}, which dynamically records the earliest timestamp at which each vertex or edge enters a $k$-core. This strategy enables substantial pruning of redundant computations. As a result, \texttt{CoreT} requires only a single pass over the query interval and achieves improved time complexity, which is linear in both the number of temporal edges within the query interval and the duration of the interval, making it highly scalable for long-term temporal analysis. Experimental results on large real-world datasets show that \texttt{CoreT} achieves up to four orders of magnitude speedup compared to the existing state-of-the-art \texttt{OTCD}, demonstrating its effectiveness and scalability for temporal $k$-core analysis.
\end{abstract}


\section{Introduction}
Many real-world networks are highly dynamic, with structures continually evolving. Such dynamics are exemplified by social networks~\cite{batagelj2011fast}, where users frequently join or leave and connections are constantly formed or dissolved. Similarly, Web graphs~\cite{vo2021efficient} undergo continual changes, with hyperlinks being created and removed as content evolves. Customer call graphs~\cite{calzada2020evaluation} also exhibit these patterns, as interactions emerge or fade when individuals expand or reduce their networks of contacts.
\emph{Temporal graphs} are particularly well-suited for modeling these evolving interactions~\cite{masuda2016guide,newman2010networks,brandes2005network}. In such graphs, each temporal edge is associated with a \emph{timestamp} that indicates when an interaction occurred, as illustrated in Figure~\ref{fig:graph}. Unlike static graphs, temporal graphs capture the dynamic evolution of relationships by recording when edges are added, modified, or removed over time. This dynamic representation makes them valuable for modeling phenomena such as financial transactions~\cite{batrancea2024financial}, temporal-aware recommendation systems~\cite{xia2022contemporary}, and infectious disease tracking~\cite{ciaperoni2020relevance}.

\begin{figure*}
    \centering
    \begin{subfigure}
        [t]{0.24\textwidth}
        \centering
        \includegraphics[width=0.9\textwidth]{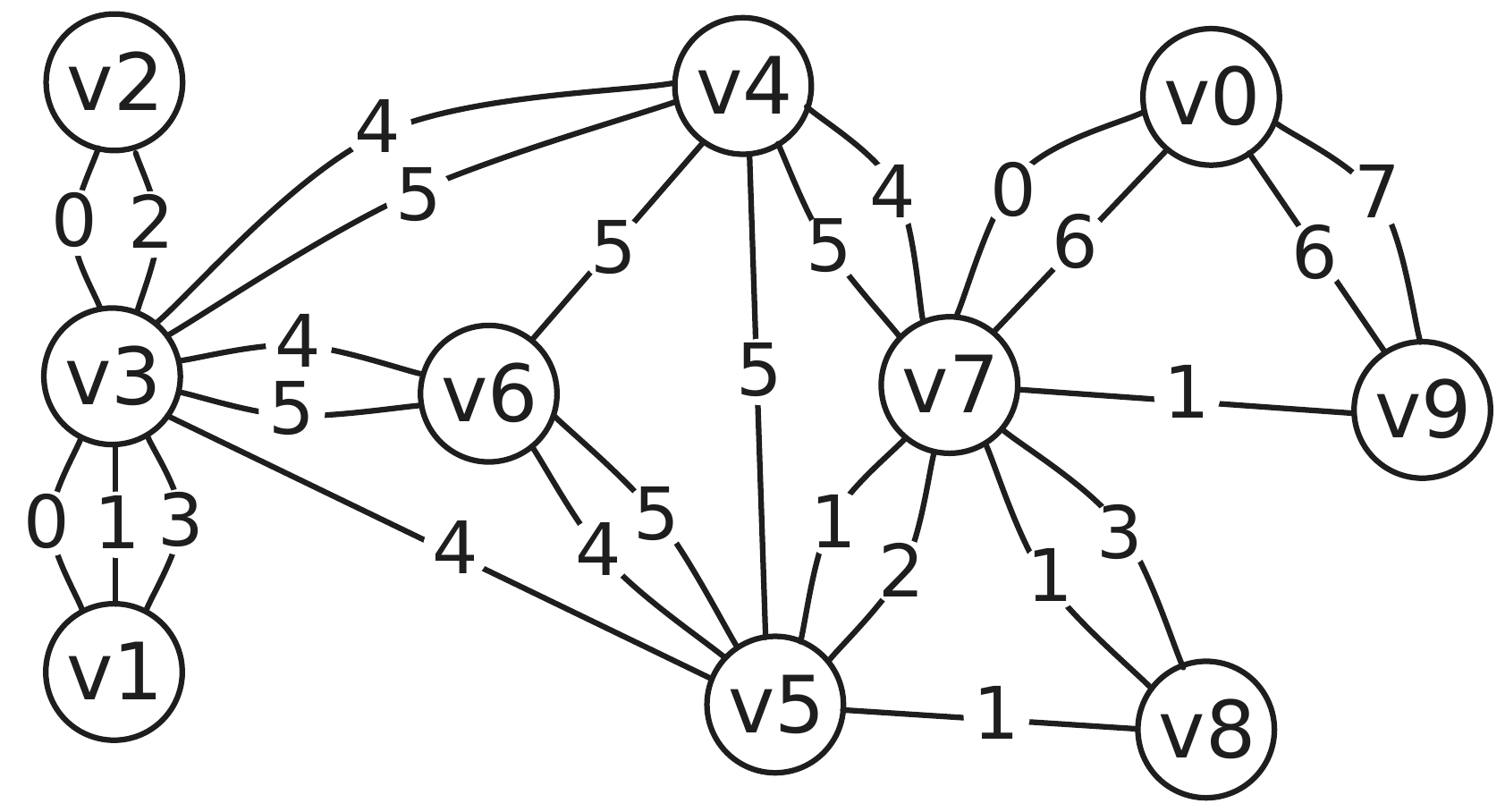}
        \caption{Temporal graph $\mathcal G$}
        \label{fig:graph}
    \end{subfigure}
    \begin{subfigure}
        [t]{0.24\textwidth}
        \centering
        \includegraphics[width=0.9\textwidth]{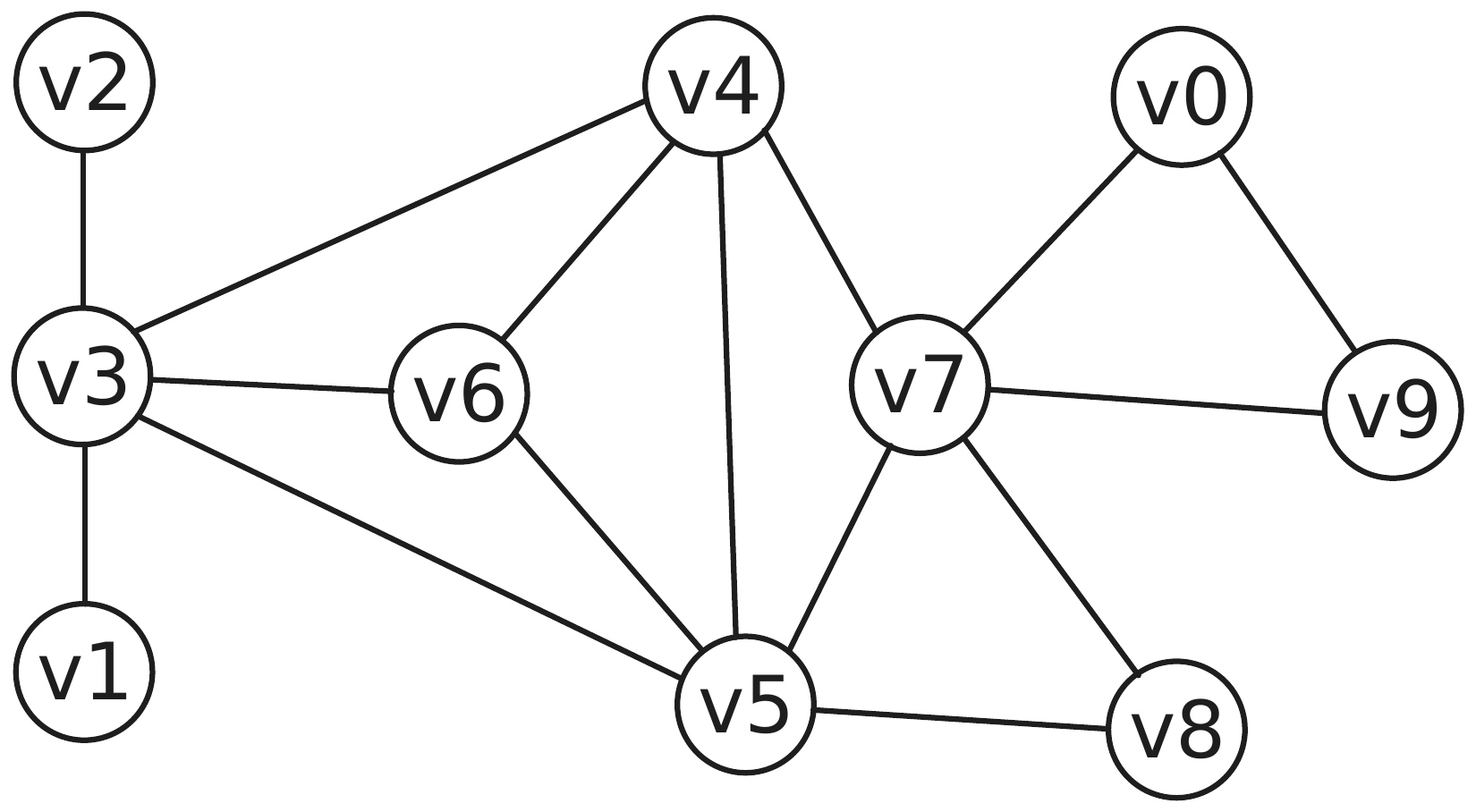}
        \caption{Underlying static graph $G$}
        \label{fig:underlying}
    \end{subfigure}
    \hfill
    \begin{subfigure}
        [t]{0.24\textwidth}
        \centering
        \includegraphics[width=0.9\textwidth]{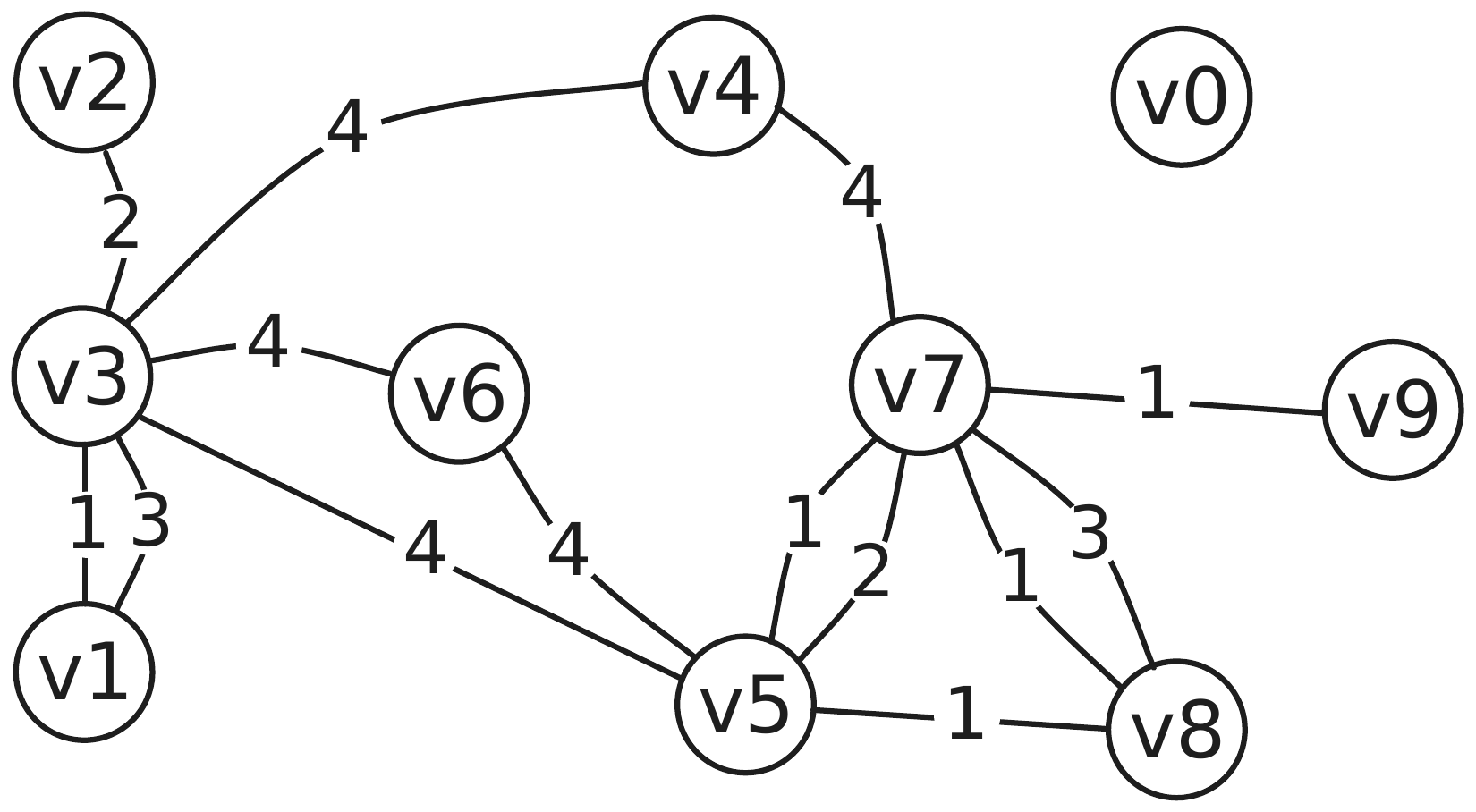}
        \caption{Projected graph over $[1,4]$}
        \label{fig:projected}
    \end{subfigure}
    \hfill
    \begin{subfigure}
        [t]{0.24\textwidth}
        \centering
    \includegraphics[width=0.9\textwidth]{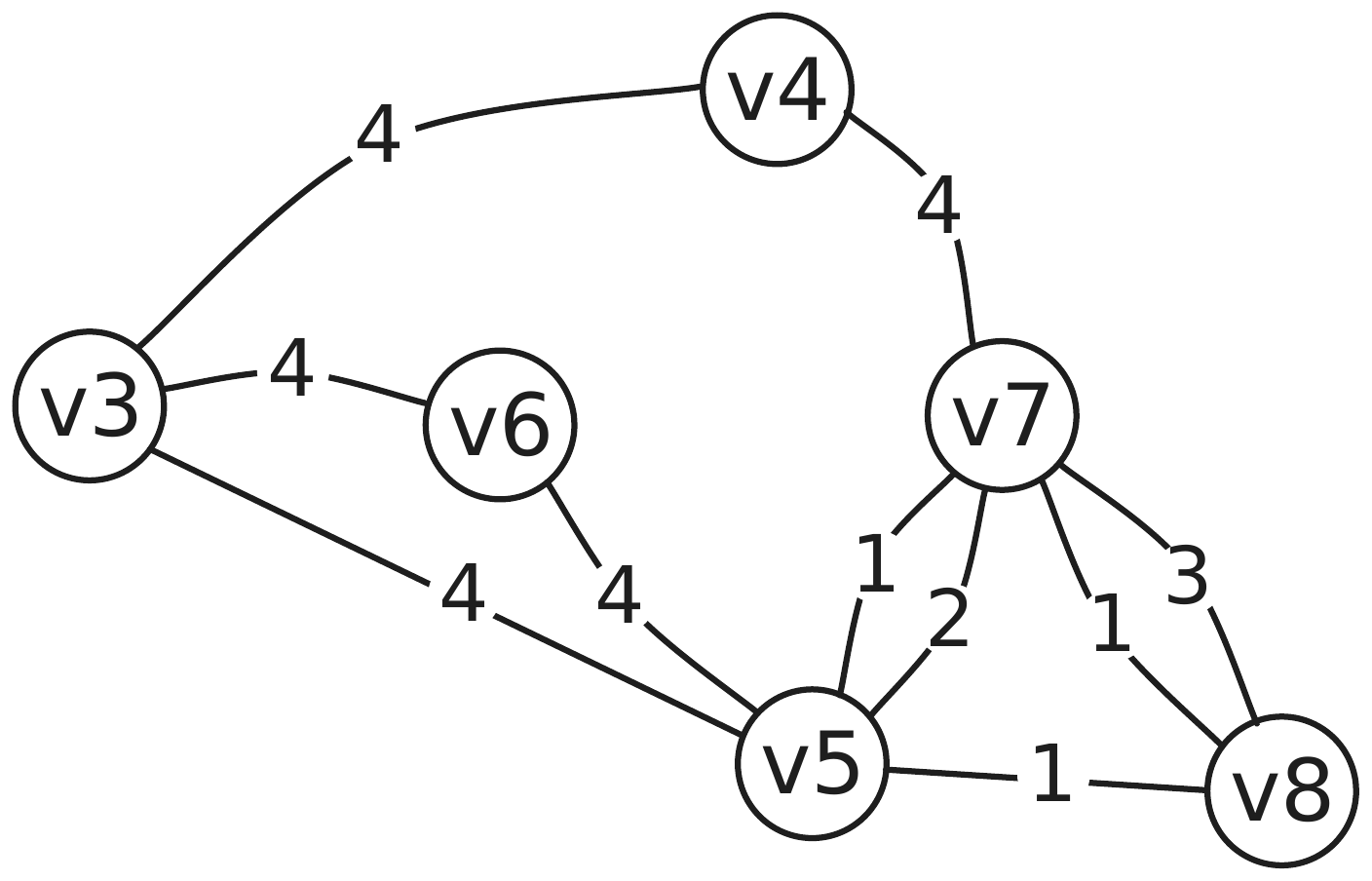}
        \caption{Temporal $2$-core in $\mathcal{G}_{[1,4]}$}
        \label{fig:core-[1,7]}
    \end{subfigure}
    \caption{A running illustrative example of key concepts in temporal $k$-core analysis: (a) the temporal graph $\mathcal{G}$, (b) its underlying static graph $G$, (c) the projected graph over $[1,4]$, and (d) the temporal $2$-core in $\mathcal{G}_{[1,4]}$.}
\end{figure*}

A fundamental task in the analysis of both static and temporal graphs is the discovery of cohesive subgraphs, which is crucial for revealing communities, detecting anomalies, and understanding functional modules within complex systems~\cite{mang2024efficient,khurana2013efficient,semertzidis2016durable,semertzidis2015timereach,wen2020efficiently}. Among the various models for cohesive subgraph mining, the $k$-core~\cite{seidman1983k-core} is particularly notable for its computational efficiency and clear structural interpretation. Formally, a $k$-core is a maximal subgraph in which every vertex connects to at least $k$ other vertices, effectively filtering out sparsely connected regions and emphasizing denser areas of the graph.
Extending cohesive subgraph mining to temporal graphs enables the identification of stable communities, persistent collaborations, and recurring behavioral patterns over time. However, incorporating the temporal dimension introduces significant complexity: the number of sub-intervals increases quadratically with the length of the time interval, substantially increasing the computational cost of enumerating all $k$-cores that appear in every interval.

In this work, we revisit the problem of temporal $k$-core query in temporal graphs, aiming to identify all distinct $k$-cores across all possible time intervals $[t_s,t_e]$ within a given query interval $[T_s,T_e]$, i.e., for all $[t_s,t_e] \subseteq [T_s,T_e]$. To address the associated computational challenges, the current state-of-the-art algorithm \texttt{OTCD}~\cite{yang2023scalable} adopts an incremental approach that exploits inclusion relationships among $k$-cores in different time intervals; specifically, the temporal $k$-core over a given interval contains the temporal $k$-core of any of its sub-intervals. Although \texttt{OTCD} reduces some redundant computations, it remains inefficient for large-scale temporal graphs due to excessive repeated processing and the combinatorial explosion in the number of intervals (e.g., inherently requiring $O(|T_e - T_s|^2)$ steps in the worst case). 

To overcome these limitations, we propose a novel algorithm \texttt{CoreT} for efficiently querying temporal $k$-cores. The key insight behind our approach is to dynamically maintain the earliest timestamp at which a vertex or edge becomes part of a $k$-core, and to leverage this information to substantially reduce unnecessary processing. Specifically, our proposed algorithm \texttt{CoreT} operates in two phases: (1) the initialization phase \texttt{CoreT\_Init} that identifies the structural states across the full time span of the query interval $[T_s,T_e]$, and (2) the iterative update phase \texttt{CoreT\_Update} that incrementally updates the graph as time advances, thereby enabling efficient querying of all valid temporal $k$-cores.
By more thoroughly exploiting the inclusion property of temporal $k$-cores, our algorithm \texttt{CoreT} avoids redundant computations and requires only a single pass over all timestamps in the query interval to identify all distinct temporal $k$-cores. This enables efficient tracking of the evolution of $k$-cores over time, even in large-scale networks. We show that this strategy allows for the efficient querying of temporal $k$-cores with the total time complexity of $O(|\mathcal E_{[T_s, T_e]}| \times |T_e - T_s|)$, where $\mathcal E_{[T_s, T_e]}$ denotes the set of temporal edges whose timestamps fall within $[T_s, T_e]$ in the given temporal graph.

Finally, we validate the effectiveness and scalability of our proposed method \texttt{CoreT} through extensive experiments on real-world datasets. The experimental results demonstrate substantial improvements in computational efficiency, with speedups of up to four orders of magnitude compared to the existing state-of-the-art algorithm \texttt{OTCD}. This significant performance gain underlines the practicality of \texttt{CoreT} for long-term temporal analyses in large-scale networks.

\section{Preliminaries}

\begin{figure*}
    \centering
    \includegraphics[width=\linewidth]{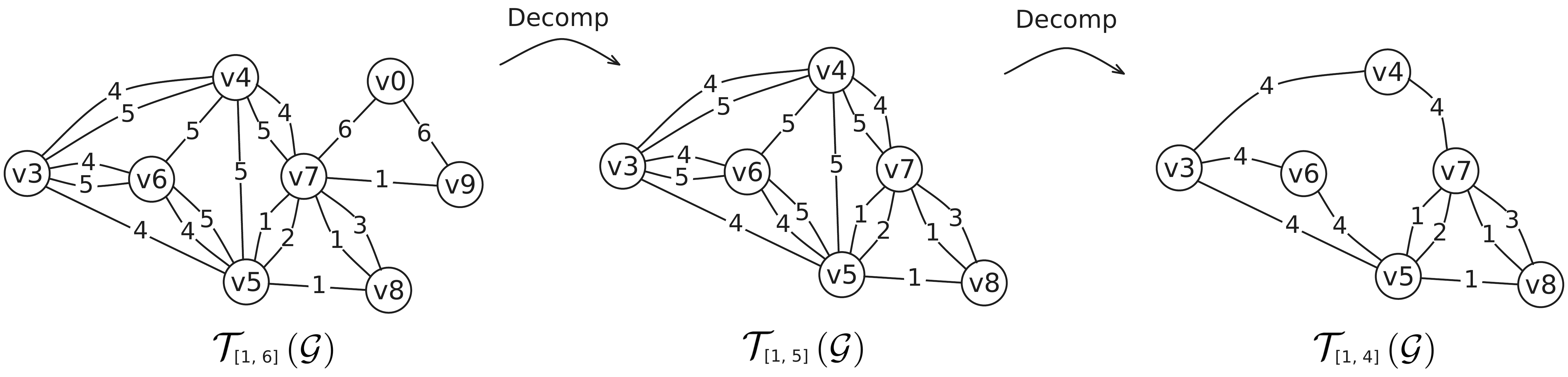}
    \caption{Illustration of \texttt{OTCD}~\cite{yang2023scalable} applied to the temporal graph in Figure~\ref{fig:graph} with $k$=2 and the query interval $[1,6]$: the procedure \texttt{Decomp} decomposes $\mathcal{T}{[1,6]}(\mathcal{G})$ to $\mathcal{T}{[1,5]}(\mathcal{G})$, and then to $\mathcal{T}_{[1,4]}(\mathcal{G})$.}
    \label{fig:decomp}
\end{figure*}

\subsection{Problem Definition}
We consider a temporal graph  $\mathcal{G}=(V, \mathcal{E})$ (which is an undirected graph with multiple parallel edges), where $V$ is the set of vertices and $\mathcal{E}$ is the set of temporal edges. Each temporal edge $(u,v,t)\in \mathcal{E}$ represents an interaction between the vertices $u$ and $v$ at timestamp $t$. Following prior work~\cite{bentert2019listing,galimberti2018mining,yang2023scalable}, we assume that timestamps are consecutive integers within a finite range $[0, t_{\max}]$, where $t_{\max}$ denotes the latest timestamp. To illustrate, see an example in Figure~\ref{fig:graph}.

Given a time interval $[t_s,t_e]$, we define the \emph{projected graph} of $\mathcal{G}$ over $[t_s,t_e]$ by $\mathcal{G}_{[ts,te]}=(V, \mathcal E_{[ts,te]})$, where $\mathcal E_{[ts,te]}$ consists of all those temporal edges with timestamps falling in the interval $[t_s,t_e]$, i.e., $\mathcal E_{[ts,te]}=\{(u,v,t) \mid (u,v,t)\in \mathcal{E}\wedge t\in[ts,te]\}$. For illustration, see an example of the projected graph (of Figure~\ref{fig:graph}) over $[1,4]$ in Figure~\ref{fig:projected}. In addition, we define the \emph{underlying static graph} (a.k.a. \emph{detemporalized graph}) of $\mathcal{G}_{[t_s,t_e]}$ by a simple undirected graph $G_{[t_s,t_e]}=(V,E_{[t_s,t_e]})$, where $E_{[t_s,t_e]}=\{(u,v)\mid (u,v,t)\in \mathcal{E}_{[t_s,t_e]}\}$. In particular, we abbreviate the underlying static graph of $\mathcal{G}$ (i.e., $G_{[0,t_{max}]}$) as $G=(V,E)$ throughout the paper. To illustrate, see the underlying static graph of Figure~\ref{fig:graph} in Figure~\ref{fig:underlying}.

We focus on one well-known cohesive subgraph structure, namely $k$-core, which is defined traditionally on static graphs~\cite{seidman1983k-core}. Specifically, given a graph $G$ and an integer $k$, a $k$-core is a maximal induced subgraph of $G$ in which every vertex has at least $k$ neighbors. Recent study~\cite{yang2023scalable} extends the $k$-core model to temporal graphs and proposes the \emph{temporal $k$-core model} (See an example in Figure~\ref{fig:core-[1,7]} for a visual illustration of the following definition).

\begin{definition}[Temporal $k$-core~\cite{yang2023scalable}]
Given a time interval $[t_s, t_e]$ and a positive integer $k$, a temporal $k$-core over $[t_s, t_e]$ in $\mathcal{G}$, denoted by $\mathcal{T}_{[t_s, t_e]}(\mathcal{G})$, is defined as a maximal induced subgraph of $\mathcal{G}_{[t_s, t_e]}$ where each vertex has at least $k$ neighbors.
\end{definition}

Compared to traditional $k$-cores (which are essentially static subgraphs), temporal $k$-cores are defined as temporal subgraphs, which captures both the temporal cohesiveness (i.e., all edges appear in a user-specified time interval $[t_s,t_e]$) and structural cohesiveness (i.e., every vertex has at least $k$ neighbors). Motivated by this, we revisit the \emph{(time-range) temporal $k$-core query} problem, which seeks to extract temporal $k$-cores from temporal graphs and has been studied in the literature~\cite{yang2023scalable}. Formally, we formulate the problem as below.

\begin{problem}[(Time-range) temporal $k$-core query~\cite{yang2023scalable}]
Given a temporal graph $\mathcal{G}$, a time interval $[T_s, T_e]$, and a positive integer $k$, the goal is to find all distinct temporal $k$-cores $\mathcal{T}_{[t_s, t_e]}(\mathcal{G})$ such that $[t_s, t_e] \subseteq [T_s, T_e]$.
\end{problem}

We note that temporal $k$-cores induced by different sub-intervals of $[T_s,T_e]$, such as $[t_s,t_e]$ and $[t_s',t_e']$, may correspond to the same subgraph of $\mathcal{G}$, i.e., $\mathcal{T}_{[t_s,t_e]}(\mathcal{G}) = \mathcal{T}_{[t_s',t_e']}(\mathcal{G})$. To reduce redundancy, the (time-range) temporal $k$-core query only returns distinct temporal $k$-cores. In addition, we define the size of query time interval $[T_s,T_e]$ as follows:
\begin{equation}
    \Delta = T_e - T_s + 1.
\end{equation}
Clearly, a temporal $k$-core query contains at most $\Delta(\Delta+1)/2$ distinct temporal $k$-cores, each of which corresponds to a sub-interval of $[T_s,T_e]$.

\smallskip
\noindent\textbf{Discussion}. We note that there exists another adaptation of the $k$-core model on temporal graphs, namely the \emph{historical $k$-core}, in the literature~\cite{yu2021querying}. Specifically, given a time interval $[t_s,t_e]$ and a positive integer $k$, a historical $k$-core is defined as a $k$-core in the detemporalized graph $G_{[t_s,t_e]}$. Thus, a historical $k$-core is essentially the detemporalized version of a temporal $k$-core, and the temporal $k$-core model generalizes the historical $k$-core model. In addition, the historical $k$-core query problem, which aims to find the historical $k$-core for a query time interval $[t_s,t_e]$, is studied in recent work~\cite{yu2021querying}. Although the proposed techniques for historical $k$-core queries can be adapted to solve the temporal $k$-core query problem, such adaptions are not efficient~\cite{yang2023scalable}. We provide further details in Section~\ref{sec:related} on related work.

\subsection{State-of-the-art Algorithms}
Given a time interval $[t_s,t_e]$, we note that the temporal $k$-core $\mathcal{T}_{[t_s,t_e]}(G)$ can be easily obtained by removing from $\mathcal{G}$ all those temporal edges with timestamps $t$ outside $[t_s,t_e]$ (i.e., $t<t_s$ or $t>t_e$) and then iteratively removing from the remaining graph those vertices (together with the incident edges) with the number of neighbors less than $k$~\cite{yang2023scalable}. We call the above procedure \texttt{Decomp}$(\mathcal{G},[t_s,t_e])$. To answer the temporal $k$-core query, one straightforward method is to enumerate all possible sub-intervals $[t_s,t_e]$ of $[T_s,T_e]$ and compute the temporal $k$-core $\mathcal{T}_{[t_s,t_e]}(\mathcal{G})$ for each sub-interval by invoking the procedure \texttt{Decomp} with the parameter $[t_s,t_e]$ on $\mathcal{G}$. To boost the efficiency, the state-of-the-art method, called \texttt{OTCD}~\cite{yang2023scalable}, explores the following nested properties.

\begin{lemma}[\cite{yang2023scalable}]
\label{lemma:nested-property} 
    Given a temporal graph $\mathcal{G}$ and two intervals $[t_s,t_e]$ and $[t_s',t_e']$ such that $[t_s,t_e]\subseteq [t_s',t_e']$, the temporal $k$-core $\mathcal{T}_{[t_s,t_e]}(\mathcal{G})$ is a subgraph of temporal $k$-core $\mathcal{T}_{[t_s',t_e']}(\mathcal{G})$ and is identical to the temporal $k$-core $\mathcal{T}_{[t_s,t_e]}(\mathcal{T}_{[t_s',t_e']}(\mathcal{G}))$, i.e.,
    \begin{equation}
        \mathcal{T}_{[t_s,t_e]}(\mathcal{G}) = \mathcal{T}_{[t_s,t_e]}(\mathcal{T}_{[t_s',t_e']}(\mathcal{G})) \subseteq \mathcal{T}_{[t_s',t_e']}(\mathcal{G}).
    \end{equation}
\end{lemma}

Lemma~\ref{lemma:nested-property} can be easily verified; thus, we omit the detailed proofs. By Lemma~\ref{lemma:nested-property}, when answering the temporal $k$-core query, one can compute a temporal $k$-core $\mathcal{T}_{[t_s,t_e]}(\mathcal{G})$ via invoking the procedure \texttt{Decomp} on a \emph{smaller} graph $\mathcal{T}_{[t_s',t_e']}(\mathcal{G})$ obtained previously (formally, \texttt{Decomp}$(\mathcal{T}_{[t_s',t_e']}(\mathcal{G}),[t_s,t_e])$) instead of the original input graph $\mathcal{G}$.

Motivated by the above observation, \texttt{OTCD} computes the temporal $k$-core $\mathcal{T}_{[t_s,t_e]}(\mathcal{G})$ for each possible sub-interval of $[T_s,T_e]$ in a certain order, and each one (excepting the first one) builds upon the previously obtained temporal $k$-core, as shown in Figure~\ref{fig:decomp}. 
Specifically, \texttt{OTCD} runs in $\Delta$ rounds. The $i$-th $(0\leq i\leq \Delta-1)$ round iteratively and sequentially computes the temporal $k$-cores induced by sub-intervals starting from the timestamp $T_s+i$, i.e., $[T_s+i,T_e], [T_s+i,T_e-1],...,[T_s+i,T_s+i]$. Following this ordering, the temporal $k$-core $\mathcal{T}_{[T_s+i,t_e]}(\mathcal{G})$ $(T_s+i\leq t_e<T_e)$ can be derived from the previously obtained one $\mathcal{T}_{[T_s+i,t_e+1]}(\mathcal{G})$ (via the the procedure \texttt{Decomp}) based on Lemma~\ref{lemma:nested-property}. In addition, the first one in the ordering $\mathcal{T}_{[T_s+i,T_e]}$ can be induced based on the one $\mathcal{T}_{[T_s+i-1,T_e]}$ obtained in the previous round (when $1\leq i\leq \Delta -1$) or can be obtained directly via \texttt{Decomp}$(\mathcal{G},[T_s,T_e])$ (when $i=0$). 
%
%
%
\texttt{OTCD} further incorporates pruning rules for pruning those redundant sub-intervals that induce identical temporal $k$-cores (details refer to~\cite{yang2023scalable}). Furthermore, it designs a data structure called Temporal Edge List (TEL) for managing the temporal graphs (in memory), which helps to boost the efficiency of edge removal and degree tracking when conducting \texttt{Decomp}. To ease the presentation, we omit the details of TEL.

\smallskip
\noindent\textbf{Time complexity}. The time cost of \texttt{OTCD} is primarily dominated by the iterative execution of \texttt{Decomp}. Specifically, it requires maintaining the TEL data structure (for tracking the degree of each vertex and deleting edges efficiently in $O(1)$) when iteratively deleting temporal edges and/or vertices during \texttt{Decomp}. Deleting an edge (resp. a vertex) from $\mathcal{G}$ takes $O(c)$ (resp. $O(\log|V|)$) for updating the TEL where $c$ is a small constant, which dominates the time cost of \texttt{Decomp}~\cite{yang2023scalable}. Recall that \texttt{OTCD} runs in $\Delta$ rounds and the $i$-th round $(0\leq i\leq \Delta-1)$ deletes $O(|\mathcal{E}_{[T_s+i,T_e]}|)$ edges and $O(|V|)$ vertices from $\mathcal{T}_{[T_s+i,T_e]} (\mathcal{G})$ in the worst case. Therefore, the time complexity of \texttt{OTCD} is bounded by $\sum_{i=T_s}^{T_e} (|V|+|\mathcal{E}_{[i,T_e]}|)\log(|V|)+c|\mathcal{E}_{[i,T_e]}|$~\cite{yang2023scalable}.

\section{Our Core-Time-based Method: \texttt{CoreT}}
\subsection{Motivation and Overview of \texttt{CoreT}}
We note that the existing method \texttt{OTCD} highly relies on the procedure \texttt{Decomp}, which will be invoked up to $O(\Delta^2)$ times and thus dominates the overall time cost of \texttt{OTCD}. However, \texttt{Decomp} is not efficient since it necessitates to maintain the TEL data structure to track vertex degrees and efficiently delete edges during the iterative removal of temporal edges or vertices. To address this limitation, we propose a new framework, namely \texttt{CoreT}, which operates independently of \texttt{Decomp}. Below, we elaborate on the details.

\begin{figure*}[t]
    \centering
    \begin{subfigure}
        [t]{0.24\textwidth}
        \centering
        \includegraphics[width=\textwidth]{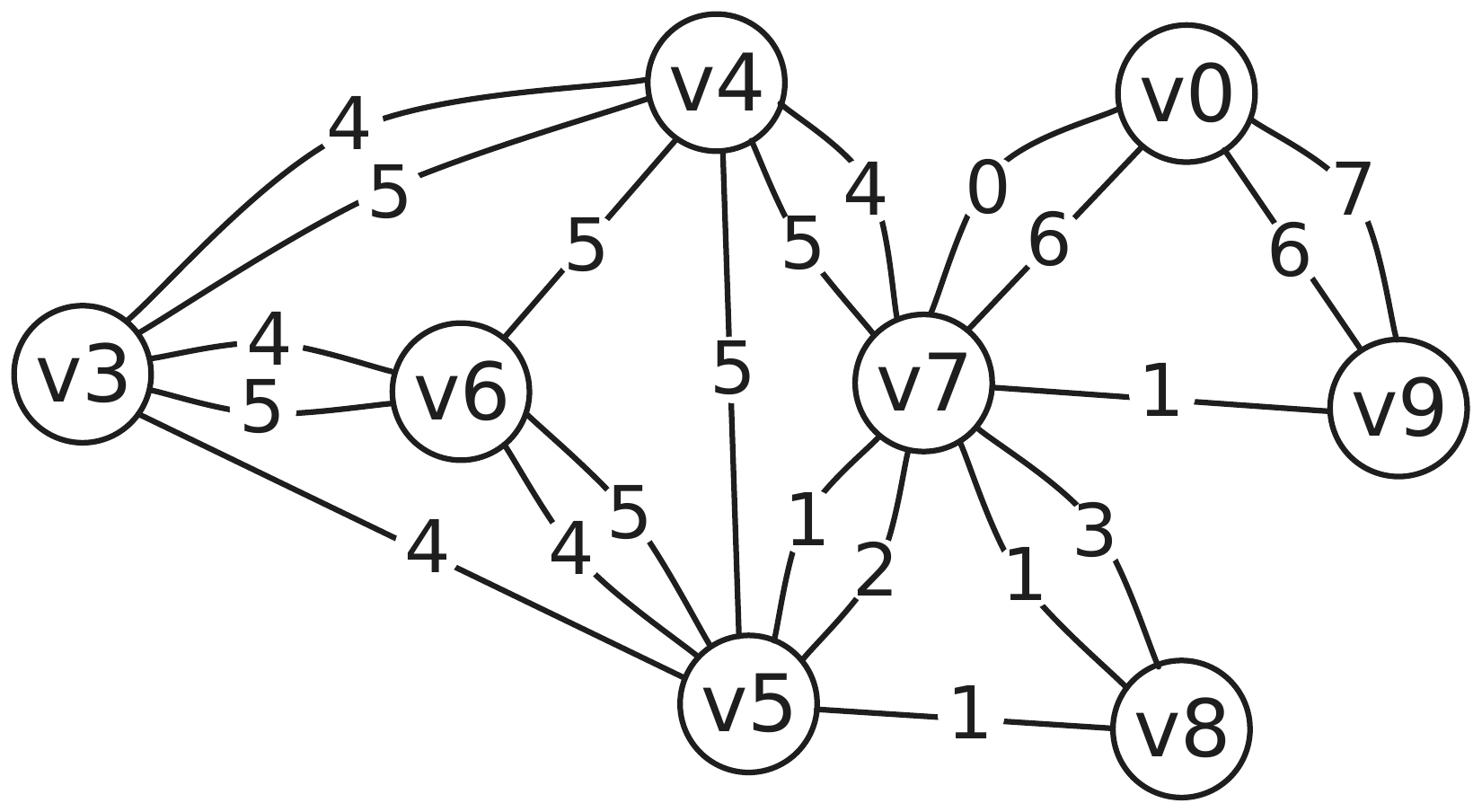}
        \caption{Temporal $2$-core in $[0,7]$}
        \label{fig:core-[0,7]}
    \end{subfigure}
    \hfill
    \begin{subfigure}
        [t]{0.24\textwidth}
        \centering
        \includegraphics[width=\textwidth]{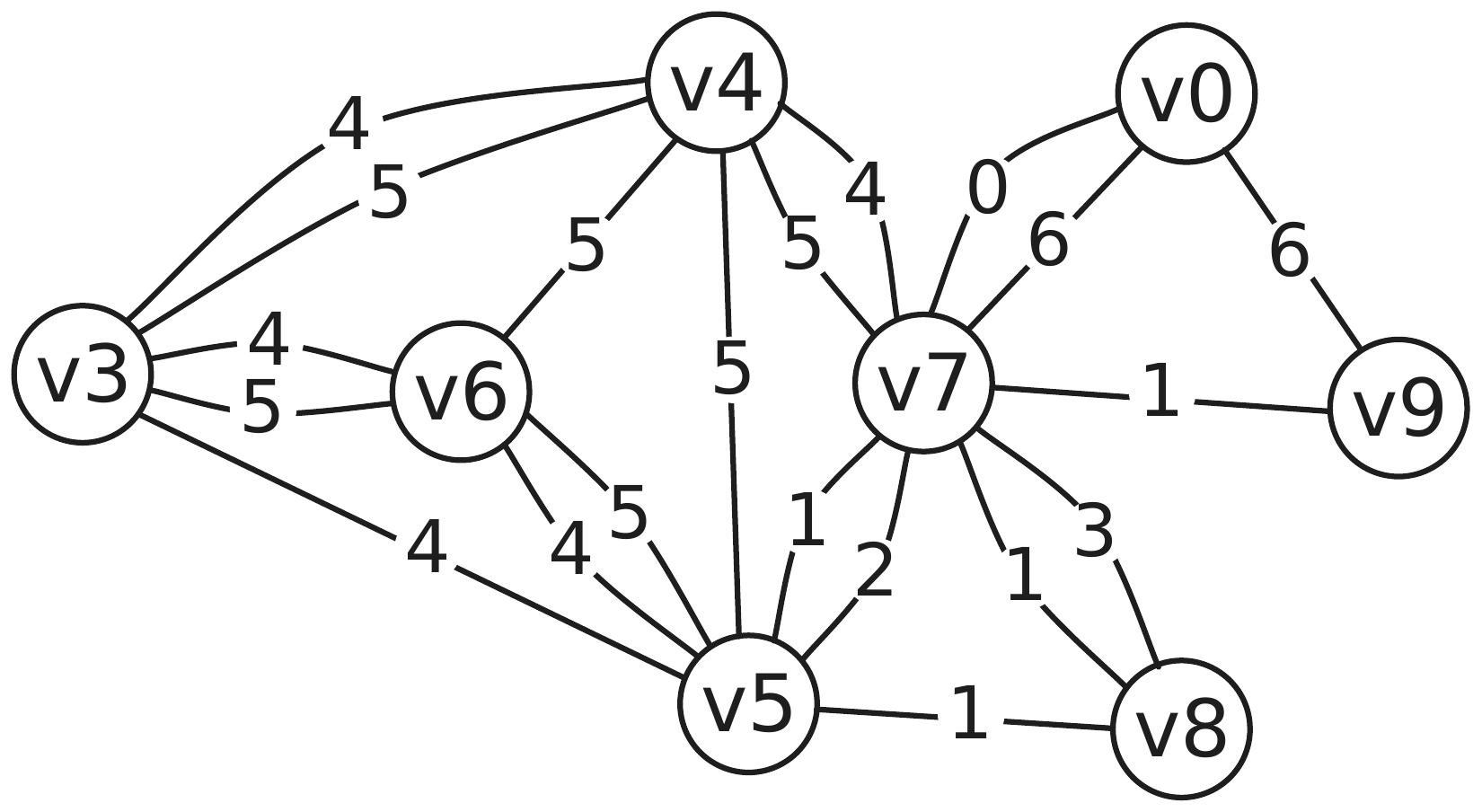}
        \caption{Temporal $2$-core in $[0,6]$}
        \label{fig:core-[0,6]}
    \end{subfigure}
    \hfill
    \begin{subfigure}
        [t]{0.24\textwidth}
        \centering
        \includegraphics[width=\textwidth]{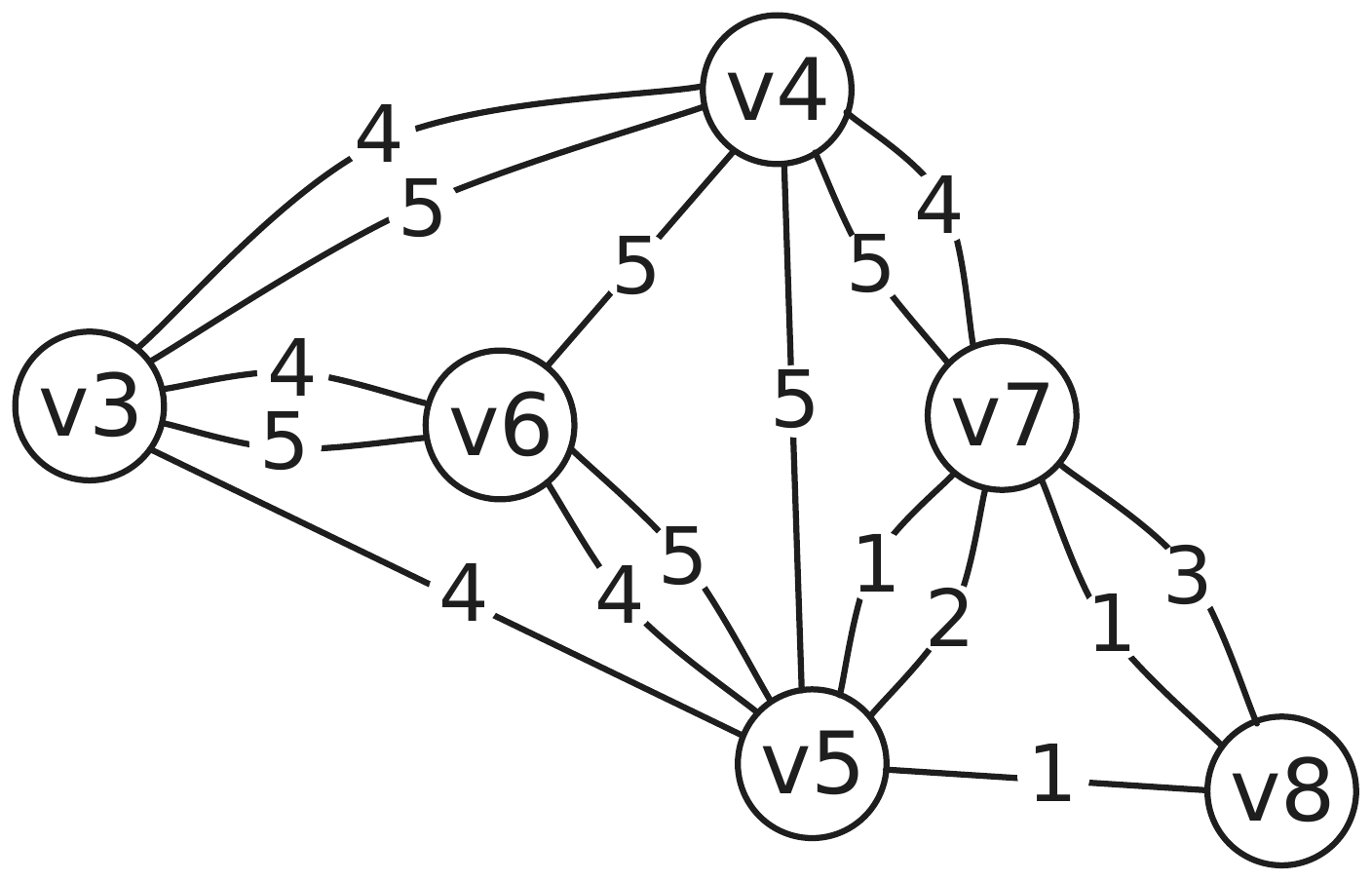}
        \caption{Temporal $2$-core in $[0,5]$}
        \label{fig:core-[0,5]}
    \end{subfigure}
    \hfill
    \begin{subfigure}
        [t]{0.24\textwidth}
        \centering
        \includegraphics[width=\textwidth]{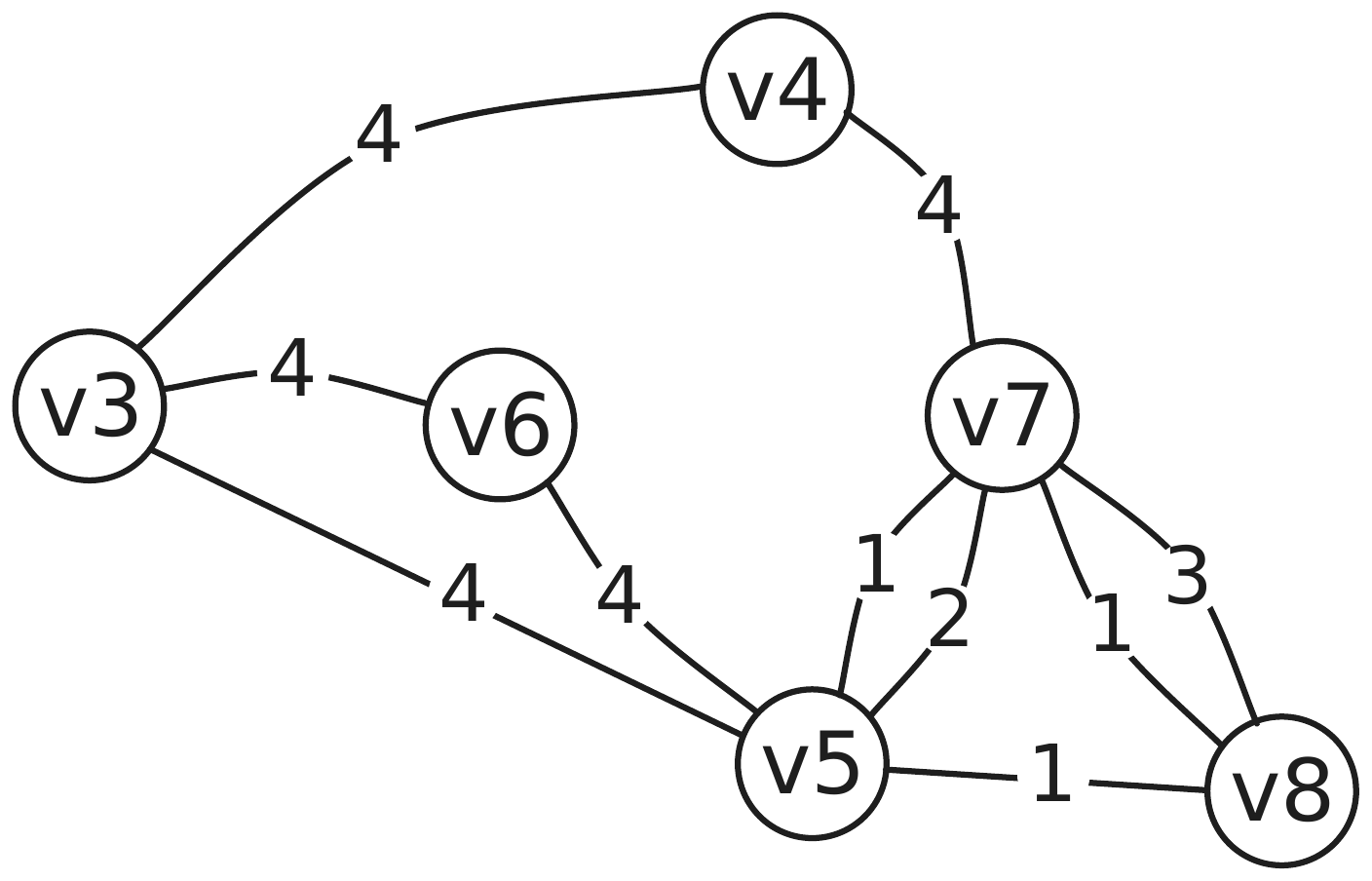}
        \caption{Temporal $2$-core in $[0,4]$}
        \label{fig:core-[0,4]}
    \end{subfigure}
    \par
    \medskip
    \begin{subfigure}
        [t]{0.24\textwidth}
        \centering
        \includegraphics[width=0.45\textwidth]{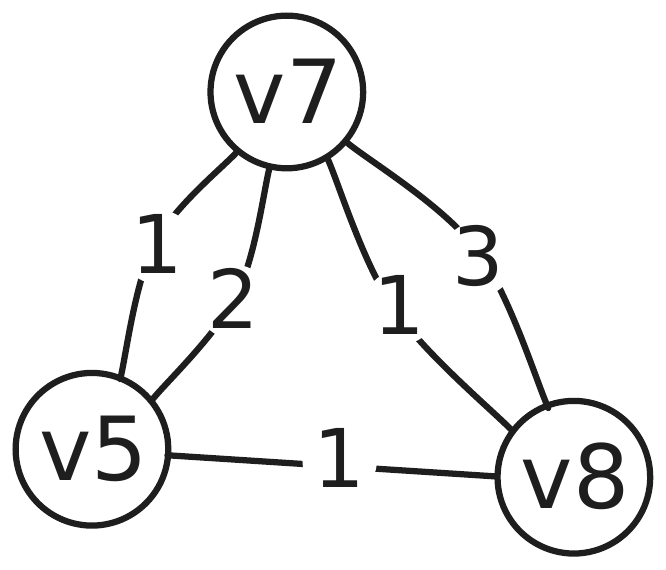}
        \caption{Temporal $2$-core in $[0,3]$}
        \label{fig:core-[0,3]}
    \end{subfigure}
    \hfill
    \begin{subfigure}
        [t]{0.24\textwidth}
        \centering
        \includegraphics[width=0.45\textwidth]{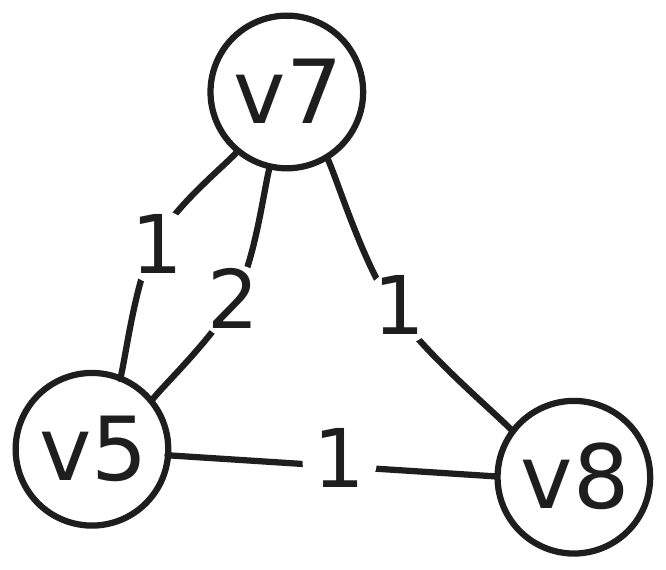}
        \caption{Temporal $2$-core in $[0,2]$}
        \label{fig:core-[0,2]}
    \end{subfigure}
    \hfill
    \begin{subfigure}
        [t]{0.24\textwidth}
        \centering
        \includegraphics[width=0.45\textwidth]{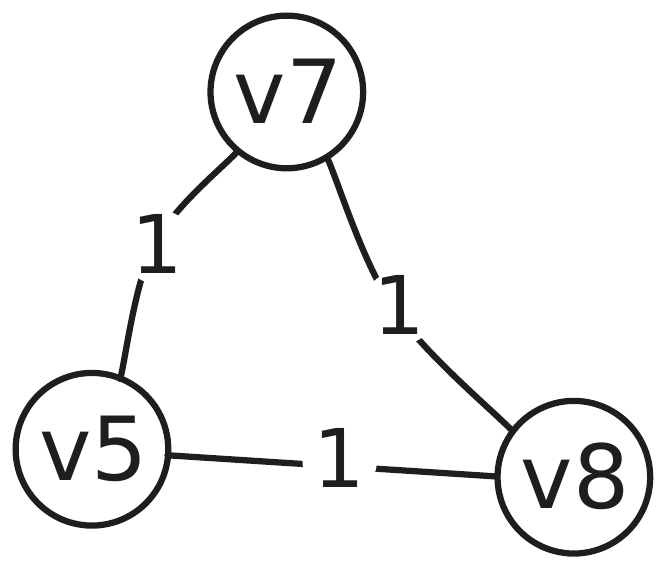}
        \caption{Temporal $2$-core in $[0,1]$}
        \label{fig:core-[0,1]}
    \end{subfigure}
    \hfill
    \begin{subfigure}
        [t]{0.24\textwidth}
        \centering
        \includegraphics[width=\textwidth]{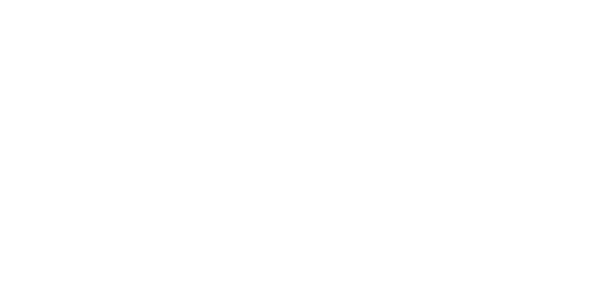}
        \caption{Temporal $2$-core in $[0,0]$}
        \label{fig:core-[0,0]}
    \end{subfigure}
    \caption{Temporal $2$-cores in intervals starting at time $0$.}
    \label{fig:main}
\end{figure*}

\begin{figure}[t]
    \centering
    \begin{subfigure}[c]{0.6\linewidth}
        \centering
        \includegraphics[width=\linewidth]{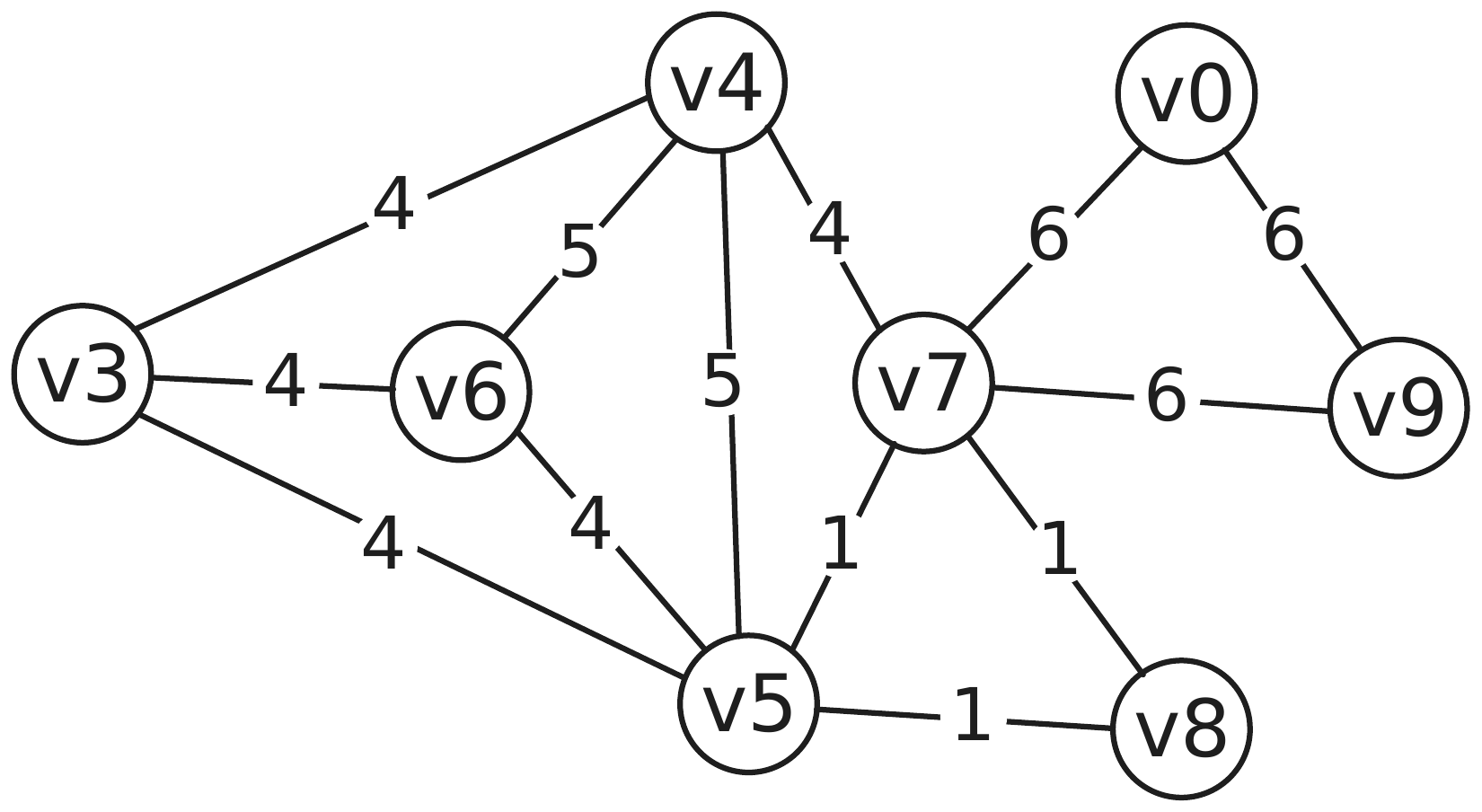}
    \end{subfigure}
    \begin{subfigure}[c]{0.3\linewidth}
        \centering
        \begin{tabular}{|c|c|}
          \hline
          v0 & 6 \\
          \hline
          v3 & 4 \\
          \hline
          v4 & 4 \\
          \hline
          v5 & 1 \\
          \hline
          v6 & 4 \\
          \hline
          v7 & 1 \\
          \hline
          v8 & 1 \\
          \hline
          v9 & 6 \\
          \hline
        \end{tabular}
    \end{subfigure}
    \caption{$2$-core times $\sigma_{0}(\cdot, \mathcal{G})$ in the initial temporal graph.}
    \label{fig:0-thres}
\end{figure}

\begin{figure}[t]
    \centering
    \begin{subfigure}[c]{0.6\linewidth}
        \centering
        \includegraphics[width=\linewidth]{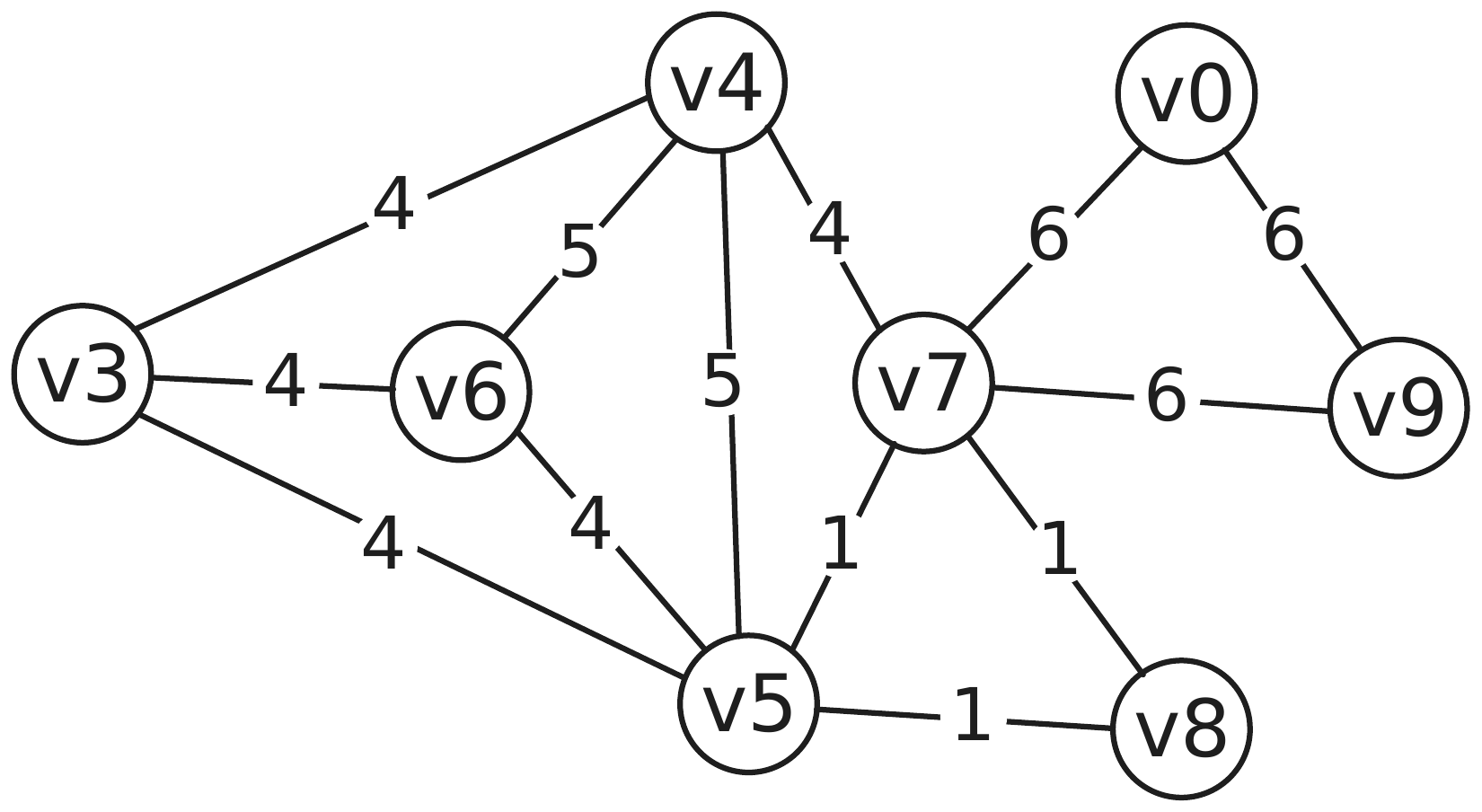}
    \end{subfigure}
    \begin{subfigure}[c]{0.3\linewidth}
        \centering
        \begin{tabular}{|c|c|}
          \hline
          v0 & 6 \\
          \hline
          v3 & 4 \\
          \hline
          v4 & 4 \\
          \hline
          v5 & 1 \\
          \hline
          v6 & 4 \\
          \hline
          v7 & 1 \\
          \hline
          v8 & 1 \\
          \hline
          v9 & 6 \\
          \hline
        \end{tabular}
    \end{subfigure}
    \caption{$2$-core time $\sigma_{1}(\cdot, \mathcal G_1)$ after removing temporal edges with timestamp $0$.}
    \label{fig:1-thres}
\end{figure}

\begin{figure}[t]
    \centering
    \begin{subfigure}[c]{0.426\linewidth}
        \centering
        \includegraphics[width=\linewidth]{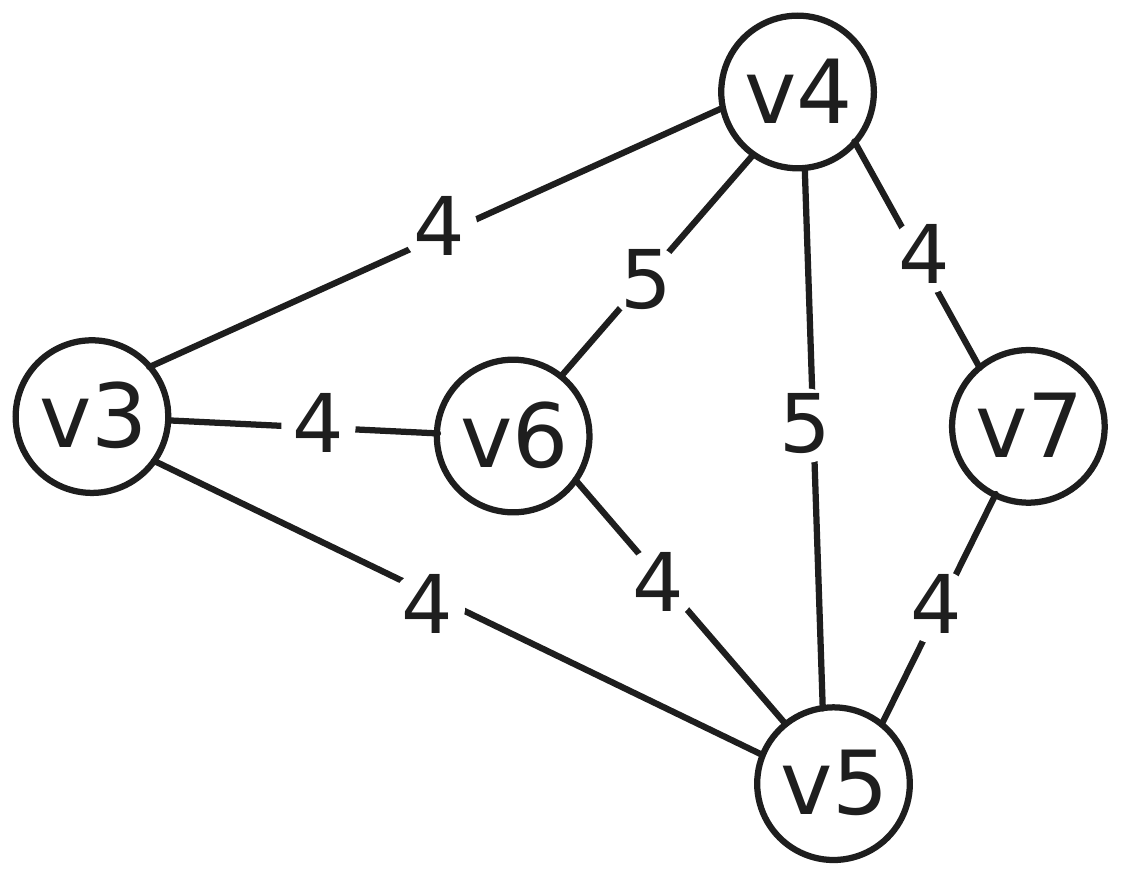}
    \end{subfigure}
    \begin{subfigure}[c]{0.3\linewidth}
        \centering
        \begin{tabular}{|c|c|}
          \hline
          v3 & 4 \\
          \hline
          v4 & 4 \\
          \hline
          v5 & 4 \\
          \hline
          v6 & 4 \\
          \hline
          v7 & 4 \\
          \hline
        \end{tabular}
    \end{subfigure}
    \caption{$2$-core time $\sigma_{1}(\cdot, \mathcal G_1)$ after removing temporal edges with timestamp $1$.}
    \label{fig:2-thres}
\end{figure}

We start with the key concept of \emph{$k$-core time}  defined below.

\begin{definition}[$k$-core time~\cite{yu2021querying}]
    Given a temporal graph $\mathcal{G}$ and a timestamp $x$ in $[0,t_{\max}]$, the $k$-core time of a vertex $v$ (resp. a detemporalized edge $(u,v)$) w.r.t. timestamp $x$, denoted by $\sigma_x(v,\mathcal{G})$ (resp. $\sigma_x(u,v,\mathcal{G})$), is defined as the minimum timestamp $t$ such that $v$ (resp. $(u,v)$) belongs to the $k$-core of the detemporalized graph $G_{[x,t]}$.\footnote{We remark that, if a vertex $v$ (resp. edge $(u,v)$) does not belong to any $k$-core of detemporalized graph $\mathcal{G}_{[x,t]}$ ($x\leq t\leq t_{\max}$), we define its $k$-core time $\sigma_x(\cdot,\mathcal{G})$ by $+\infty$ for better consistency.}
\end{definition}

We observe that, (\textbf{Motivation 1}) \emph{given the $k$-core times $\sigma_x(\cdot,\mathcal{G})$, all temporal $k$-cores with the same start time $x$, i.e., $\{\mathcal{T}_{[x,t_e]}(\mathcal{G})\mid x\leq t_e\leq t_{\max}\}$, can be efficiently retrieved from $\mathcal{G}$}. Specifically, a temporal $k$-core $\mathcal{T}_{[x,t_e]}(\mathcal{G})$ consists of all those vertices with $k$-core time $\sigma_x$ at most $t_e$ (i.e., $\{v\in V\mid \sigma_x(v,\mathcal{G})\leq t_e\}$) and all temporal edges whose corresponding detemporalized edge has $k$-core time at most $t_e$ and whose timestamp is at most $t_e$
(i.e., $\{(u,v,t)\in\mathcal{E}\mid \sigma_x(u,v,\mathcal{G})\leq t_e \text{ and } t\leq t_e\}$). Formally, we summarize it in the following lemma.

\begin{lemma}
\label{lemma:core-time-query}
    Let $\mathcal{G}$ be a temporal graph and $0\leq x\leq t_e\leq t_{\max}$, we have
    \begin{eqnarray}
        \mathcal{T}_{[x,t_e]}(\mathcal{G}) = (\{v\in V\mid \sigma_x(v,\mathcal{G})\leq t_e\},\nonumber\\ \{(u,v,t)\in\mathcal{E}\mid \sigma_x(u,v,\mathcal{G})\leq t_e \text{ and } t\leq t_e\}).
    \end{eqnarray}
\end{lemma}
\begin{proof}
    This can be easily verified based on the definition of $k$-core time. 
\end{proof}

Based on Lemma~\ref{lemma:core-time-query}, we can efficiently answer the temporal $k$-core query by utilizing a series of $k$-core times, i.e., $\sigma_{T_s},\sigma_{T_s+1},...,\sigma_{T_e}$. We further observe that (\textbf{Motivation 2}) \emph{$\sigma_{T_s}$ can be computed efficiently in $O(|E|)$}, where $E$ is the edge set of the detemporalized graph of $\mathcal{G}$ (details will be discussed in Section~\ref{subsec:coreT-init}), and (\textbf{Motivation 3}) \emph{$\sigma_{x+1}$ can be obtained efficiently based on $\sigma_x$ where $T_s\leq x\leq T_e-1$} (details will be discussed in Section~\ref{subsec:coreT-update}). 

Motivated by all the above, we propose a new framework called \texttt{CoreT}, which is summarized in Algorithm~\ref{alg:enumeration}. Specifically, it includes two phrases, namely \emph{Phase I: initialization} and \emph{Phase II: Iterative update}, as below.

\smallskip
\noindent\textbf{Phase I: Initialization}. Our method \texttt{CoreT} begins by computing the $k$-core times $\sigma_{T_s}(\cdot)$. Specifically, it first refines the input graph $\mathcal{G}$ to $\mathcal{G}_0=\mathcal{G}_{[T_s,T_e]}$ (Line 1), since all temporal $k$-cores within $[T_s,T_e]$ are subgraphs of $\mathcal{G}_{[T_s,T_e]}$. Then, for each vertex and edge in the detemporalized graph of $\mathcal{G}_0$, it computes the $k$-core time $\sigma_{T_s}(\cdot,\mathcal{G}_0)$ via the procedure of \texttt{CoreT\_Init} (Line 2). The technical details will be discussed in Section~\ref{subsec:coreT-init}. Finally, according to $\sigma_{T_s}(\cdot,\mathcal{G}_0)$ and Lemma~\ref{lemma:core-time-query}, it outputs all those temporal $k$-cores with start time $T_s$, i.e., $\{\mathcal{T}_{[T_s,t_e]}(\mathcal{G})\mid T_s\leq t_e\leq T_e\}$ via the procedure \texttt{CoreT\_List} (Line 3, Lines 7-11).

\smallskip
\noindent\textbf{Phase II: Iterative update}. \texttt{CoreT} iteratively performs $\Delta-1$ iterations, where the $t$-th $(1\leq t\leq \Delta-1)$ round aims to output all those temporal $k$-cores with start timestamp $T_s+t$, i.e., $\{\mathcal{T}_{[T_s+t,t_e]}\mid T_s+t\leq t\leq T_e\}$ (Line 6).
\underline{First}, we note that the input temporal graph $\mathcal{G}_0$ (which contains all temporal $k$-cores to be output) can be progressively refined during the iterations. 
Specifically, at the $t$-th iteration, we refine $\mathcal{G}_{t-1}$ (obtained either from a former iteration or initially as $\mathcal{G}_0$) to a smaller one $\mathcal{G}_{t}$ by dropping some edges that cannot appear in any temporal $k$-core to be yielded in this round (Line 5).
The rationale is as follows. Since $\mathcal{G}_{t-1}$ (initially $\mathcal{G}_0$) contains all temporal $k$-cores with start timestamp $T_s+t-1$, it also contains all temporal $k$-cores with start timestamp $T_s+t$, according to Lemma~\ref{lemma:nested-property}. Thus, the refined input graph $\mathcal{G}_{t}$ contains all temporal $k$-cores with start timestamp $T_s+t$.  
\underline{Second}, we obtain $\sigma_{T_s+t}(\cdot,\mathcal{G}_t)$ based on $\sigma_{T_s+t-1}(\cdot,\mathcal{G}_{t-1})$ (Line 5). Specifically, instead of computing $\sigma_{T_s+t}(\cdot,\mathcal{G}_t)$ from scratch (i.e., via \texttt{CoreT\_Init}$(T_s+t,\mathcal{G}_t)$), we propose to update $\sigma_{T_s+t-1}(\cdot)$ to $\sigma_{T_s+t}(\cdot)$ by considering the removal of those temporal edges with timestamp $T_s+t-1$. 
We will discuss (1) how to refine the input graph $\mathcal{G}_{t-1}$ efficiently and (2) how to obtain $\sigma_{T_s+t}(\cdot)$ based on $\sigma_{T_s+t-1}(\cdot)$ in Section~\ref{subsec:coreT-update}.

\smallskip
\noindent\textbf{Illustrative example}. We illustrate the execution of \texttt{CoreT} using the temporal graph $\mathcal{G}$ in Fig.~\ref{fig:graph}. We consider a full-span query with start timestamp $T_s = 0$ and end timestamp $T_e = 7$, and set $k = 2$. Accordingly, we let $\mathcal{G}_0 = \mathcal{G}$.
In \textbf{Phase I}, we compute the initial $2$-core times $\sigma_0(\cdot, \mathcal{G}_0)$ for all vertices and edges. Figure~\ref{fig:graph} presents the temporal graph, and Figure~\ref{fig:0-thres} illustrates the computed core times, where edge labels indicate the $2$-core times of edges and vertex core times are listed in the accompanying table. For example, since $\sigma_0(v) = 4$, vertex $v$ appears in the $2$-core of $\mathcal{G}_{[0, t]}$ only when $t \geq 4$.
Using these core times, we enumerate all temporal $2$-cores with start time $0$, such as $\mathcal{T}_{[0,4]}(\mathcal{G})$ shown in Figure~\ref{fig:core-[0,4]}, which includes all vertices and edges with $2$-core time at most $4$.
In \textbf{Phase II}, we iteratively compute the core times. At each iteration $t$, we first refine $\mathcal{G}_t$ by removing all temporal edges with timestamp $t-1$ from $\mathcal{G}_{t-1}$, then update the core times in the detemporalized graph accordingly, and finally enumerate all temporal $k$-cores with start time $t$ by using the updated core times.
Figures~\ref{fig:1-thres} and~\ref{fig:2-thres} depict the core times obtained after the first and second iterations, respectively. The updated core times at each step enable efficient enumeration of temporal $2$-cores beginning at the corresponding start timestamps $t = 1$ and $t = 2$.

\begin{algorithm}[t]
    \caption{Our core-time-based method: \texttt{CoreT}}
    \label{alg:enumeration} 
    \KwIn{Temporal graph $\mathcal{G}=(V, \mathcal{E})$, a time interval $[T_s,T_e]$ and a positive integer $k$}
    \KwOut{All distinct temporal $k$-cores $\mathcal{T}_{[t_s,t_e]}(\mathcal{G})$ such that $[t_s,t_e]\subseteq [T_s,T_e]$}
    \tcc{\textbf{Phase I}: Initialization}
    $\mathcal{G}_0\leftarrow \mathcal{G}_{[T_s,T_e]}$\tcp*[r]{Refine $\mathcal{G}$}
    $\sigma_{T_s}(\cdot,\mathcal{G}_0)\leftarrow$\texttt{CoreT\_Init}$(T_s,\mathcal{G}_0)$\tcp*[r]{Obtain $\sigma_{T_s}$}
    
    \texttt{CoreT\_List}$(\sigma_{T_s}(\cdot,\mathcal{G}_0),\mathcal{G}_0)$\;
    
    \tcc{\textbf{Phase II}: Iterative update}
    \For{$t \gets 1, 2, \dots, \Delta -1 $}{ 
        $\sigma_{T_s+t}                            (\cdot,\mathcal{G}_t),\mathcal{G}_t\leftarrow$\texttt{CoreT\_Update}$(\sigma_{T_s+t-1}(\cdot,\mathcal{G}_{t-1}),$\ $\mathcal{G}_{t-1})$\;
        \texttt{CoreT\_List}$(\sigma_{T_s+t},\mathcal{G}_t)$;
    }
    
    \SetKwBlock{Enum}{Procedure \texttt{CoreT\_List}$(\sigma_x(\cdot,\mathcal{G}_i),\mathcal{G}_i)$}{}
    \Enum{
        \For{$t_e\leftarrow x,x+1,\cdots, T_e$}{
            $V^*\gets \{u\in V_i\mid \sigma_x(u,\mathcal{G}_i) \leq t_e\}$\;
            $\mathcal{E}^*\gets \{(u,v,t)\in\mathcal{E}_i\mid \sigma_{x}(u,v,\mathcal{G}_i)\leq t_e \ \text{and} \  t \le t_e\}$\;
            \textbf{yield} $\mathcal{T}_{[x,t_e]}(\mathcal{G})\gets (V^*,\mathcal{E}^*)$\;
        }
    }
\end{algorithm}

\begin{figure}[t]
    \centering
    \includegraphics[width=\linewidth]{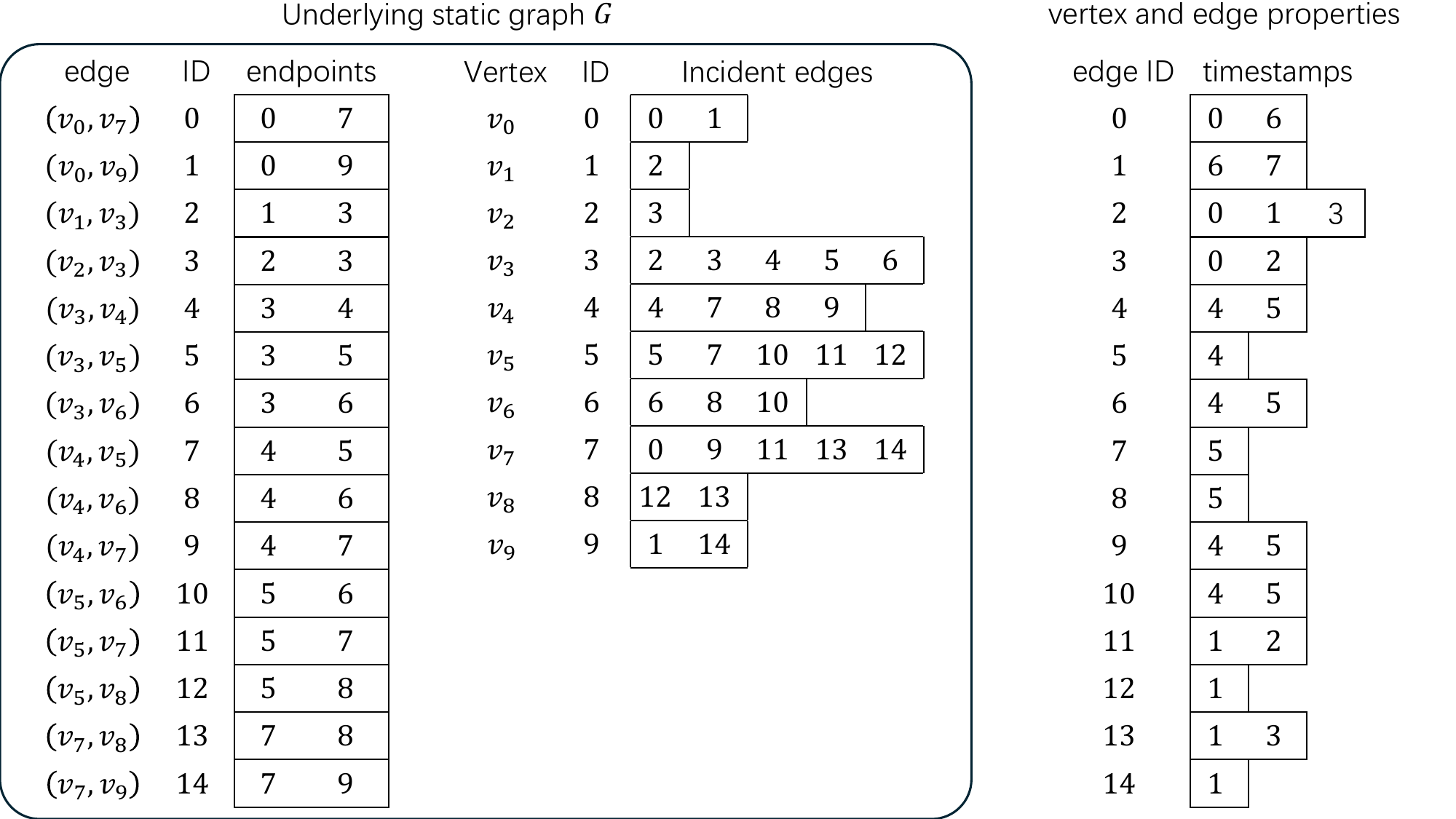}
    \caption{Illustration of temporal graph representation.}
    \label{fig:graph_represent}
\end{figure}

\smallskip
\noindent\textbf{Details of the implementation.} We remark that the choice of data structure for storing a temporal graph in memory significantly affects the performance of graph algorithms, both in theory and in practice. Instead of adopting TEL for storing temporal graphs (as in the existing work~\cite{yang2023scalable}), we use a new data structure that is both simple and better suited to our proposed method. Specifically, for a temporal graph $\mathcal{G}$, we store its detemporalized graph $G$ (by using classic adjacency lists) and a list of detemporalized edges where each edge is assigned with a list of values \emph{sorted in ascending} (which could correspond to endpoints, timestamps and core times). To illustrate, see an example in Figure~\ref{fig:graph_represent}. Clearly, the space complexity is $O(|V|+|\mathcal{E}|)$.

\smallskip
\noindent\textbf{Time complexity analysis}. 
We analyze the complexity of our proposed \texttt{CoreT}. In particular, \texttt{CoreT} consists of two phases: Phase I includes one \texttt{CoreT\_Init} and one \texttt{CoreT\_List} call, while Phase II consists of \(\Delta = T_e - T_s + 1\) iterations, each invoking both \texttt{CoreT\_Update} and \texttt{CoreT\_List}. Each of the procedures \texttt{CoreT\_Init}, \texttt{CoreT\_Update}, and \texttt{CoreT\_List} runs in \(O(|\mathcal{E}_{[T_s, T_e]}|)\). The complexity of \texttt{CoreT\_List} is analyzed below, while the detailed descriptions and analyses of \texttt{CoreT\_Init} and \texttt{CoreT\_Update} are provided in Sections~\ref{subsec:coreT-init} and~\ref{subsec:coreT-update}, respectively. Thus, the overall time complexity of \texttt{CoreT} is $O(\Delta\times |\mathcal E_{[T_e-T_s]})$.

We analyze the complexity of \texttt{CoreT\_List}. We note that for each call, the temporal $k$-cores enumerated by \texttt{CoreT\_List} over sub-intervals $[x, t_e]$ form a monotonically increasing sequence with respect to edge and vertex inclusion, i.e.,
\[
\mathcal{T}_{[x,x]}(\mathcal{G}) \subseteq \mathcal{T}_{[x,x+1]}(\mathcal{G}) \subseteq \cdots \subseteq \mathcal{T}_{[x,T_e]}(\mathcal{G}).
\]  
Accordingly, \texttt{CoreT\_List} maintains and updates the sets \(V^*\) and \(\mathcal{E}^*\) by incrementally adding newly included edges and vertices at each step. The total cost of one \texttt{CoreT\_List} call at iteration \(i\) is linear in the size of the temporal graph \(\mathcal{G}_i\), bounded by \(O(|\mathcal{E}_{[T_s, T_e]}|)\).

\smallskip
\noindent \textbf{Remark}. To ensure that only unique temporal $k$‑cores are returned, both the state-of-the-art algorithm \texttt{OTCD} and our proposed algorithm \texttt{CoreT} exploit the fact that each distinct tightest time interval (TTI) uniquely identifies a temporal $k$‑core. In particular, the TTI refers to the minimal sub-interval $[t_s', t_e']$ within the original query interval $[t_s, t_e]$ that contains all temporal edges of a given temporal $k$‑core $\mathcal{T}_{[t_s, t_e]}(\mathcal{G})$. That is, $t_s'$ and $t_e'$ correspond to the earliest and latest timestamps among the edges in the core, respectively. This compact interval captures the precise temporal span during which the $k$‑core structure is maintained. Thus, once a temporal $k$‑core is generated, it suffices to check whether its TTI has already been encountered to ensure the uniqueness of the result. Checking the uniqueness of TTIs requires at most  $O(|T_e - T_s|^2)$ time. Under the assumption that each timestamp is associated with at least one edge, this cost is strictly smaller than $|E_{[T_s, T_e]}|$. It does not affect the overall complexity of either \texttt{OTCD} or \texttt{CoreT}. Therefore, we omit the implementation details for duplicate core removal in this paper.

\subsection{\texttt{CoreT\_Init}: Core Time Computation}
\label{subsec:coreT-init}

Consider the process of computing $k$-core times $\sigma_{T_s}(\cdot,\mathcal{G}_0)$ for vertices and edges in the detemporalized graph of $\mathcal{G}_0$ (denoted by $G_0$). We start with the following definitions.

\begin{definition}[$k$-degree time]
    Given a temporal graph $\mathcal{G}$ and a timestamp $x$ in $[0,t_{\max}]$, the $k$-degree time of a vertex $v$ w.r.t timestamp $x$, denoted by $d_x(v,\mathcal{G})$, is the minimum timestamp $t$ such that $v$ has degree at least $k$ in the detemporalized graph $G_{[x,t]}$.\footnote{We remark that if a vertex $v$ has the degree less than $k$ in $G_{x,t_{\max}}$, we denote its $k$-degree time $d_x(v,\mathcal{G})$ by $+\infty$.} 
\end{definition}

\begin{definition}[support time]
    Given a temporal graph $\mathcal{G}$ and a detemporalized edge $(u,v)$ in $G$, the support time of $(u,v)$, denoted by $Sup(u,v,\mathcal{G})$, is the minimum timestamp $t$ such that $(u,v,t)$ exists, i.e.,
    \begin{equation}
        Sup(u,v,\mathcal{G})=\min\{t \mid (u,v,t)\in\mathcal{E}\}.
    \end{equation}
\end{definition}

In the following, we first discuss how to compute the support times $Sup(\cdot,\mathcal{G}_0)$ and degree times $d(\cdot,\mathcal{G})$ for $\mathcal{G}_0$, and then leverage the degree times $d(\cdot,\mathcal{G})$ to compute the core times $\sigma_{T_s}(\cdot,\mathcal{G}_0)$.

\smallskip
\noindent\textbf{Computation of support time and degree time}. The support times $Sup(\cdot,\mathcal{G}_0)$ of all detemporalized edges in $G_0$ can be obtained directly in $O(|E|)$ (i.e., by taking the first timestamp from the assigned list for each edge) based on our graph representation. Since the incident edges of each vertex $v$ are pre-sorted by support time, the $k$-degree time $d_{T_s}(v, \mathcal{G}_0)$ can be computed in $O(1)$ time by directly accessing the $k$-th smallest value.
Therefore, the overall computation requires $O(|E_0|)$ time, which is \emph{linear} in the input size.

\begin{algorithm}[t]
    \caption{\texttt{Core\_Init}}
    \label{alg:decomp} \KwIn{Temporal graph $\mathcal{G}_0$}
    \KwOut{$k$-core time $\sigma_{T_s}(\cdot,\mathcal{G}_0)$ for each vertex and edge in $G_0$}
    Compute support times $Sup(u,v,\mathcal{G}_0)$ of all edges\;
    Compute $k$-degree times $d_{T_s}(\cdot,\mathcal{G}_0)$ of all vertices\; 
    $t\gets +\infty$, $\mathcal{G}'\gets \mathcal{G}_0$\; 
    \lForEach{vertex $v$ in $V_0$}{$d_{T_s}(v,\mathcal{G}')\gets d_{T_s}(v,\mathcal{G}_0)$} 
    \While{$V'$ is not empty}{
        $v^*\gets \arg\max_{v\in V'} d_{T_s}(v,\mathcal{G}')$\;
        $t\gets \min\{t,d_{T_s}(v^*,\mathcal{G}')\}$\;
        $\sigma_{T_s}(v^*,\mathcal{G}_0)\gets t$\;
        $V'\gets V'\backslash \{v^*\}$\;
        \tcc{Update $k$-degree time }
        \ForEach{neighbor $v$ of $v^*$}{
                $d_{T_s}(\cdot,\mathcal{G}')\gets$ the $k$-th smallest one among $\{Sup(v^*,v,\mathcal{G}_0)\mid v\in V' \& (v,v^*)\in E_0\}$;
        }
    }
    \ForEach{edge $(u,v)$ in $E_0$}{
        $\sigma_{T_s}(u,v,\mathcal{G}_0)\gets \max\{\sigma_{T_s}(u,\mathcal{G}_0),\sigma_{T_s}(v,\mathcal{G}_0),Sup(u,v,\mathcal{G}_0)\}$;
    }
    \Return $\sigma_{T_s}(\cdot,\mathcal{G}_0)$;
\end{algorithm}

\smallskip
\noindent\textbf{Computation of core time}. We observe that the $k$-core times $\sigma_{T_s}(\cdot,\mathcal{G}_0)$ of vertices can be obtained via an iterative peeling procedure, as summarized in Algorithm~\ref{alg:decomp}. Specifically, we iteratively remove from the current temporal graph $\mathcal{G}'=(V',\mathcal{E}')$ a vertex $v$ with the largest $k$-degree time (Lines 3-12). The $k$-core time $\sigma_{T_s}(v,\mathcal{G}_0)$ is set to the minimum value between the $k$-degree time of the previously removed vertex (initially $+\infty$) and that of the currently removed one (Line 7). When a vertex is removed from the current graph $\mathcal{G}'$, we need to update the $k$-degree times of its neighbors (Lines 10-12). Clearly, the $k$-degree times $d{T_s}(\cdot, \mathcal{G}')$ of non-neighbors of $v$ remain unchanged after $v$ is removed. We remark that if $v$ has less than $k$ neighbors in the remaining graph $\mathcal{G}'$, its $k$-degree time is set by $+\infty$.
Then, the $k$-core time of each edge $(u,v)$, i.e., $\sigma_{T_s}(u,v,\mathcal{G}_0)$ is set to the minimum value of the $k$-core times of $u$ and $v$ (Lines 12-13).
Below, we formally prove the correctness of Algorithm~\ref{alg:decomp}.

\begin{lemma}
    Algorithm~\ref{alg:decomp} finds the $k$-core times $\sigma_{T_s}(\cdot,\mathcal{G}_0)$ of all vertices and edges in the detemporalized graph $G_0$.
\end{lemma}
\begin{proof}
    We first show that Algorithm~\ref{alg:decomp} correctly computes the $k$-core times of all vertices at Lines 5-11.
    Let $w$ be a vertex in $V_0$, $\sigma_{T_s}(w,\mathcal{G}_0)$ be its $k$-core time, and $t$ be the value to be assigned to $\sigma_{T_s}(w,\mathcal{G}_0)$ at Line 8. We need to show that $\sigma_{T_s}(w,\mathcal{G}_0)$ is equal to $t$. In general, there are two cases. \underline{First}, we show that $\sigma_{T_s}(w,\mathcal{G}_0)\leq t$. Consider the iterative peeling procedure at Lines 5-11. Let $v'$ be the first vertex to be removed from $G'$ such that $d_{T_s}(v',G')=t$ (at Line 6). Note that such a vertex $v'$ exists obviously (which could be either $w$ or one vertex removed before). Clearly, $\mathcal{G}'_{[T_s,t]}$ is a subgraph of the temporal $k$-core $\mathcal{T}_{[T_s,t]}(\mathcal{G}_0)$ since every vertex in $\mathcal{G}'$ has $k$-degree time $d_{T_s}(\cdot,\mathcal{G}')$ at most $t$, and thus, by definition, every vertex has at least $k$ neighbors in $\mathcal{G}'_{[T_s,t]}$. Then, we can derive that the $k$-core time $\sigma_{T_s}(w,\mathcal{G}_0)$ is at most $t$ (since $\mathcal{T}_{[T_s,t]}(\mathcal{G}_0)$ contains $w$ based on the above discussion). \underline{Second}, we show that $\sigma_{T_s}(w,\mathcal{G}_0)\geq t$. Consider the iteration in which $w$ is removed from $V'$. We can easily derive that $d_{T_s}(w,\mathcal{G}')\geq t$ by Line 7. Furthermore, we have $d_{T_s}(w,\mathcal{G}_0)\geq d_{T_s}(w,\mathcal{G}')\geq t$ based on the update process at Lines 10-11 (note that the degree times are non-decreasing after the updates). Clearly, we easily derive that $\sigma_{T_s}(w,\mathcal{G}_0)\geq d_{T_s}(w,\mathcal{G}_0)$ based on definition. Then, we have $\sigma_{T_s}(w,\mathcal{G}_0)\geq t$. \underline{Finally}, combining the above, we conclude that $\sigma_{T_s}(w,\mathcal{G}_0)= t$.

    We then show that Algorithm~\ref{alg:decomp} returns the $k$-core times of all detemporalized edges in $\mathcal{G}_0$ by contradiction. Consider an edge $(u,v)$. In general, there are two cases. \underline{First}, if $\sigma_{T_s}(u,v,\mathcal{G}_0)<\sigma_{T_s}(u,\mathcal{G}_0)$ or $\sigma_{T_s}(u,v,\mathcal{G}_0)<\sigma_{T_s}(v,\mathcal{G}_0)$, the detemporalized $k$-core $G^*_{[T_s,t]}$ (where $t=\sigma_{T_s}(u,v,\mathcal{G}_0)$ and $G^*=G_0$) cannot contain $(u,v)$ since either $u$ or $v$ is not in $G_{[T_s,t]}$, which leads to a contradiction. \underline{Second}, if $\sigma_{T_s}(u,v,\mathcal{G}_0)<Sup(u,v,\mathcal{G}_0)$, the detemporalized $k$-core $G^*_{[T_s,t]}$ cannot contain $(u,v)$ based on the definition of support time, which leads to a contradiction.
\end{proof}

\smallskip
\noindent\textbf{Time complexity analysis}. We analyze the complexity of \texttt{CoreT\_Init} (Algorithm~\ref{alg:decomp}).
First, we compute the support times $Sup(u,v,\mathcal{G}_0)$ for all edges $(u,v)$ and the $k$-degree times for all vertices in $O(|E_0|)$ time (Lines 1-2). During the iterative vertex removal process (Lines 5-11), we maintain buckets indexed by support times, assigning each vertex to the bucket corresponding to its current $k$-degree time. At each step, we remove the vertex with the largest $k$-degree time by scanning the buckets in decreasing order. After removing a vertex $v^*$, we update the $k$-degree times of its neighbors in $O(d)$ time, where $d$ is the degree of $v^*$ (Lines 10-11). Finally, the core time for each edge $(u,v)$ is computed as the maximum of the core times of its endpoints and its own support time, which takes $O(1)$ time per edge and $O(|E_0|)$ in total. Thus, the overall procedure runs in $O(|E_0|)$ time.

\subsection{\texttt{CoreT\_Update}: Iterative Core Time Update}
\label{subsec:coreT-update}
Consider the process of refining $\mathcal{G}_{t-1}$ to $\mathcal{G}_t$ and computing $\sigma_{T_s+t} (\cdot,\mathcal{G}_{T_s+t})$ based on $\sigma_{T_s+t-1}(\cdot,\mathcal{G}_{T_s+t-1})$ at the $t$-th $(1\leq t\leq \Delta-1)$ iteration of \texttt{CoreT}. Recall that $\mathcal{G}_t$ is a subgraph of $\mathcal{G}_{t-1}$ and must contain all temporal $k$-cores with start timestamp $T_s+t$. Typically, a smaller $\mathcal{G}_t$ is preferred for improved practical performance.
We first introduce the following pruning rule for refining $\mathcal{G}_{t-1}$ to $\mathcal{G}_t$.

\begin{itemize}
    \item \textbf{Pruning rule.} All temporal edges with timestamp $T_s+t-1$ can be removed from $\mathcal{G}_{t-1}$ at the $t$-th iteration of \texttt{CoreT}.
\end{itemize}

The correctness of this rule is obvious, as any temporal $k$-core to be output at the $t$-th iteration has the start timestamp equal to $T_s+t$ and thus cannot contain those temporal edges. As a result, we propose to refine $\mathcal{G}_{t-1}$ by dropping all those temporal edges with the timestamp $T_s+t-1$, formally,
\begin{equation}
    \mathcal{G}_t:=(\mathcal{G}_{t-1})_{[T_s+t,T_e]}.
\end{equation}
To compute $\sigma_{T_s+t}(\cdot,\mathcal{G}_{t})$, we have the following observations.

\begin{align*}
\forall v \in V_t :\quad \sigma_{T_s+t-1}(v,\mathcal G_{t-1}) & \le \sigma_{T_s+t-1}(v,(\mathcal G_{t-1})_{[t_s+t,T_e]}) \\
& = \sigma_{T_s+t-1}(v,\mathcal{G}_t) \\
& = \sigma_{T_s+t}(v,\mathcal{G}_t).
\end{align*}
\vspace{-20pt}
\begin{align*}
\forall (u,v) \in E_t:\quad &\sigma_{T_s+t-1}(u,v,\mathcal{G}_{t-1})\\
&\leq \sigma_{T_s+t-1}(u,v,(\mathcal{G}_{t-1})_{[T_s+t,T_e]}) \\
&=\sigma_{T_s+t-1}(u,v,\mathcal{G}_t) \\
&= \sigma_{T_s+t}(u,v,\mathcal{G}_t).
\end{align*}

These equations follow directly from the definitions; thus, we omit the detailed proofs.

\smallskip
\noindent\textbf{Overview of \texttt{Core\_Update}}. Motivated by the above observations, we propose a procedure for refining $\mathcal{G}_{t-1}$ and simultaneously computing $\sigma_{T_s+t}(\cdot,\mathcal{G}_{t})$. The idea is to remove all temporal edges $(u,v,T_s+t-1)$ from $\mathcal{G}_{t-1}$, and then update the $k$-core times $\sigma_{T_s+t-1}(\cdot,\mathcal{G}_{t-1})$ for certain vertices and edges. Clearly, after removing all such temporal edges with timestamp $T_s+t-1$, the $k$-core times $\sigma_{T_s+t-1}(\cdot,\mathcal{G}_{t-1})$ are equal to $\sigma_{T_s+t}(\cdot,\mathcal{G}_t)$ based on the above equations. The remaining problem is how to update $\sigma_{T_s+t-1}(\cdot,\mathcal{G}_{t-1})$ after removing a temporal edge or a vertex from $\mathcal{G}_{t-1}$. Below, we elaborate on the details.

\smallskip
\noindent\textbf{Update rules}.
We start with the following two lemmas. 
\begin{lemma}
\label{lemma:update-rule-1}
    Let $x$ be a timestamp in $[0,t_{\max}]$ and $v$ be a vertex in $V$. The $k$-core time of $v$, i.e., $\sigma_x(v,\mathcal{G})$, is equal to the $k$-th smallest value among the $k$-core times of its incident edges, i.e., $\{\sigma_x(v,u,\mathcal{G})\mid (u,v)\in E\}$.
\end{lemma}
\begin{proof}
Let \( t = \sigma_x(v, \mathcal{G}) \) be the $k$-core time of vertex \( v \) with respect to timestamp \( x \). By definition, \( v \) belongs to the $k$-core of \( G_{[x, t]} \), which implies that it has at least \( k \) incident edges that also lie in the $k$-core of \( G_{[x, t]} \). For each such edge \( (u,v) \), it holds that \( \sigma_x(u,v, \mathcal{G}) \leq t \). Therefore, the $k$-th smallest core time among all incident edges of \( v \) is at most \( t \). If this value were strictly smaller than \( t \), then \( v \) would belong to the $k$-core of some \( G_{[x, t']} \) with \( t' < t \), contradicting the minimality of \( \sigma_x(v, \mathcal{G}) \).
\end{proof}

\begin{lemma}
\label{lemma:update-rule-2}
    Let $(u,v)$ be an edge in $E$ and $x$ be a timestamp in $[0,t_{\max}]$. We have
    \begin{equation}
        \sigma_x(u,v,\mathcal{G}) = \max\{Sup(u,v,\mathcal{G}),\sigma_x(u,\mathcal{G}),\sigma_x(v,\mathcal{G})\}.
    \end{equation}
\end{lemma}
\begin{proof}
For $(u,v)$ to be included in the $k$-core of $\mathcal{G}_{[x,t]}$, the following conditions must hold simultaneously:
\begin{enumerate}
    \item The edge $(u,v)$ exists within the interval $[x,t]$, i.e., $t \geq Sup(u,v,\mathcal{G})$, where $Sup(u,v,\mathcal{G})$ denotes the earliest timestamp of $(u,v)$ in $\mathcal{G}$.
    \item Both vertices $u$ and $v$ belong to the $k$-core of $\mathcal{G}_{[x,t]}$, so $t \geq \sigma_x(u,\mathcal{G})$ and $t \geq \sigma_x(v,\mathcal{G})$.
\end{enumerate}
Hence, $\sigma_x(u,v,\mathcal{G}) = \max \{ Sup(u,v,\mathcal{G}), \sigma_x(u,\mathcal{G}), \sigma_x(v,\mathcal{G}) \}$.
\end{proof}

These lemmas clarify the dependencies between the $k$-core times of vertices and edges. When a set of temporal edges $\mathcal{R}=\{(u,v,x)\in \mathcal{E}_{t-1}\mid x=T_s+t-1\}$ is removed from $\mathcal{G}_{t-1}$, the following updates are required, according to the above lemmas: (1) the support time of each edge in $\mathcal{R}$, (2) the $k$-core time of each edge in $\mathcal{R}$, and (3) the $k$-core time of each vertex incident to an edge in $\mathcal{R}$. Furthermore, these updates may trigger further updates due to changes in core times. A naive approach would be to iteratively perform updates until the $k$-core times stabilize, but this can lead to redundant computations due to repeated updates. To improve efficiency, we propose an update ordering.

\smallskip
\noindent\textbf{Update ordering}. Consider the iterative update process described above. For consistency, let $\widetilde{\sigma}_{T_s+t}(\cdot, \mathcal{G}_t)$ denote the tentative $k$-core time of each vertex and edge in $\mathcal{G}_t$ during the update process. Initially, $\widetilde{\sigma}_{T_s+t}(\cdot,\mathcal{G}_t)$ is set to $\sigma_{T_s+t-1}(\cdot,\mathcal{G}_{t-1})$. 
We propose to process updates in \emph{descending order} of $\widetilde{\sigma}_{T_s+t}(\cdot,\mathcal{G}_t)$. Specifically, among the vertices and edges that will trigger updates, we prioritize those associated with the largest value of $\widetilde{\sigma}_{T_s+t}(\cdot, \mathcal{G}_t)$. The rationale is that \emph{subsequent updates triggered by other vertices or edges will not affect these elements again}.
To illustrate, consider the case where vertex $u$ has the current largest tentative $k$-core time $t^*$ and will trigger updates. Then, we can deduce that any edge $(u,v)$ (that will trigger updates) in $\mathcal{G}_t$ has the tentative $k$-core time at most $t^*$. In addition, for any edge $(u,v)$ (that will trigger updates), the tentative $k$-core time of $v$ is also at most $t^*$; otherwise, by Lemma~\ref{lemma:update-rule-2}, the tentative $k$-core time of $(u,v)$ would be greater than $t^*$. Therefore, after processing the updates triggered by $u$, any edge $(u,v)$ will have the tentative $k$-core time at most $t^*$ based on Lemma~\ref{lemma:update-rule-2}. In subsequent updates triggered by $(u,v)$, it is clear from Lemma~\ref{lemma:update-rule-1} that the tentative $k$-core time of $u$ remains unchanged.
To ease the presentation, we omit the formal proof.

Motivated by all the above, we summarize the proposed method \texttt{Core\_Update} in Algorithm~\ref{alg:update}. In particular, we first update the $k$-core time of those removed edges (Lines 2-6). Then, we iteratively conduct the updates triggered by a vertex or an edge in the descending order of the tentative $k$-core time (Lines 7-22). Finally, we return the refined graph $\mathcal{G}_t$ and the corresponding $k$-core time ${\sigma}_{T_s+t}(\cdot,\mathcal{G}_t)$. The correctness can be guaranteed based on the above discussion.

\begin{algorithm}[t]
    \caption{\texttt{Core\_Update}}
    \label{alg:update}
    \KwIn{Temporal graph $\mathcal{G}_{t-1}$, $\sigma_{T_s+t-1}(\cdot,\mathcal{G}_{t-1})$}
    \KwOut{$\mathcal{G}_t$ and $\sigma_{T_s+t}(\cdot,\mathcal{G}_t)$}
    
    $S \gets \emptyset,  \widetilde{\sigma}_{T_s+t-1}(\cdot,\mathcal{G}_t)\gets\sigma_{T_s+t-1}(\cdot,\mathcal{G}_{t-1})$\;
    
    \ForEach{edge $(u,v,T_s+t-1) \in \mathcal{E}_{t-1}$}{
        Compute the support time $Sup(u,v,\mathcal{G}_1)$\;
        \If{$Sup(u,v,\mathcal{G}_1) > \widetilde{\sigma}_{T_s+t-1}(u,v,\mathcal{G}_t)$}{
            $\widetilde{\sigma}_{T_s+t-1}(u,v,\mathcal{G}_t) \gets Sup(u,v,\mathcal{G}_1)$\;
            $S \gets S \cup \{(u,v)\}$\;
        }
    }
    
    \While{$S \neq \emptyset$}{
        $s \gets \arg\max_{x \in S} \widetilde{\sigma}_{T_s+t-1}(x,\mathcal{G}_t)$\;
        $S \gets S \setminus \{s\}$\;
        $\sigma_{T_s+t}(s,\mathcal{G}_t) \gets \widetilde{\sigma}_{T_s+t-1}(s,\mathcal{G}_t)$\;

        \uIf{$s$ is a vertex}{
            \ForEach{edge $(s,v)$ incident to $s$}{
                $\delta \gets \max \{ \widetilde{\sigma}_{T_s+t-1}(s,\mathcal{G}_t), \widetilde{\sigma}_{T_s+t-1}(v,\mathcal{G}_t) \}$\;
                \If{$\widetilde{\sigma}_{T_s+t-1}(s,v,\mathcal{G}_t) < \delta$}{
                    $\widetilde{\sigma}_{T_s+t-1}(s,v,\mathcal{G}_t) \gets \delta$\;
                    $S \gets S \cup \{(s,v)\}$\;
                }
            }
        }
        \ElseIf{$s$ is an edge}{
            \ForEach{vertex $v$ in edge $s$}{
                $\delta \gets$ $k$-th smallest $\widetilde{\sigma}_{T_s+t-1}(\cdot,\mathcal{G}_t)$ among edges incident to $v$\;
                \If{$\widetilde{\sigma}_{T_s+t-1}(v,\mathcal{G}_t) < \delta$}{
                    $\widetilde{\sigma}_{T_s+t-1}(v,\mathcal{G}_t) \gets \delta$\;
                    $S \gets S \cup \{v\}$\;
                }
            }
        }
    }
    \Return $(\mathcal{G}_t)_{[T_s+t,T_e]}$, ${\sigma}_{T_s+t}(\cdot,\mathcal{G}_t)$;
\end{algorithm}

\smallskip
\noindent\textbf{Time complexity analysis}. We consider the complexity of \texttt{CoreT\_Updte} (Algorithm~\ref{alg:update}).
In Lines 2–6, the algorithm scans all temporal edges $(u,v,t)$ that are active at time $T_s + t - 1$. Each such edge is visited once to compute its current support time and, if necessary, update its tentative core time. This step takes time linear in the number of temporal edges at time $T_s + t - 1$,
Lines 7–22 perform an iterative propagation of updated core times throughout the graph. This process is divided into two branches depending on whether a vertex or an edge is extracted from the candidate set $S$. According to the update ordering enforced in Algorithm~\ref{alg:update}, each vertex or edge enters and leaves the candidate set once.
When a vertex $s$ is selected (Lines 11–16), the algorithm examines each of its incident edges. If the degree of $s$ is $d$, this step takes $O(d)$ time. Since each edge is inspected at most once, the total cost for all vertex-triggered updates is bounded by $O(|\mathcal E_{[T_e-T_s]}|)$.
Prior to invoking \texttt{Core\_Update}, we maintain for each vertex a sorted list of its incident edges in ascending order of tentative $k$-core time. Since these core times are updated in a non-decreasing manner during the propagation process, each updated edge can be reinserted into the appropriate position in the list via a forward scan. This incremental maintenance incurs $O(d)$ time per vertex, where $d$ is the vertex degree, leading to an overall cost of $O(|\mathcal{E}_{[T_s,T_e]}|)$ per iteration. This preprocessing allows the $k$-th smallest tentative core time to be retrieved in $O(1)$ time during Line 18. Consequently, each edge-triggered update in Lines 17–22 is performed in constant time. Thus, the total cost of edge-triggered updates is also bounded by $O(|\mathcal E_{[T_e-T_s]}|)$.
Combining all components, the overall time complexity of \texttt{Core\_Update} is linear in the number of edges, that is, $O(|\mathcal E_{[T_e-T_s]}|)$.

\section{Experiments}
In this section, we conduct experiments to verify both the efficiency and effectiveness of our proposed algorithm \texttt{CoreT} in comparison with the state-of-the-art algorithm \texttt{OTCD}~\cite{yang2023scalable}. All experiments are performed on a Linux machine equipped with an Intel Xeon 2.60 GHz CPU and 256 GB RAM. The algorithms are implemented in C++ and compiled using g++ -O3.

\smallskip
\noindent\textbf{Dataset}.
Our experiments are conducted on temporal graphs from SNAP\footnote{https://snap.stanford.edu/data/} and KONECT\footnote{http://konect.cc/networks/}, covering all datasets in~\cite{yang2023scalable}.

\begin{table}[t]
    \centering
    \caption{Statistics of all datasets}
    \label{tab:dataset}
    \begin{tabular}{lccc}
    \toprule
        \textbf{Datasets} & \textbf{$|V|$} & \textbf{$|\mathcal E|$} & \textbf{$t_{\max}$} \\
    \midrule
        CollegeMsg (CM) & 1,862 & 59,835 & 58,910 \\
        sx-superuser (SU) & 159,481 & 144,339 & 1,437,198 \\
        email-Eu-core (EE) & 966 & 332,334 & 207,879 \\
        sx-mathoverflow (MO) & 21,980 & 506,550 & 505,783 \\
        sx-askubuntu (AU) & 127,047 & 964,437 & 960,865 \\
        wiki-talk-temporal (WT) & 1,120,716 & 7,833,140 & 7,375,041 \\
        dblp-coauthor (DC) & 1,824,702 & 29,487,744 & 77 \\ 
        flickr-growth (FG) &  2,245,501 &  33,139,584 & 134 \\
        wikipedia-growth (WG) & 1,870,710 & 39,953,145 & 2,198 \\
    \bottomrule
    \end{tabular}

\end{table}

\smallskip
\noindent \underline{\textbf{Exp-1: Time comparison with the baseline.}}
To systematically evaluate the impact of both the core threshold~$k$ and the query interval size on algorithmic efficiency, we consider two sets of parameter values. Specifically, the core threshold $k$ varies from 2 to 20, and the query interval is configured as $[T_s, T_e] = [0,\ \alpha \times t_{\max}]$, where $\alpha \in \{0.2,\,0.4,\,0.6,\,0.8,\,1.0\}$. For each pair of \((k,\alpha)\), we record the execution/response time in milliseconds (i.e., $10^{-3}$ seconds).

\begin{figure}[t]
    \centering
    \includegraphics[width=0.9\linewidth]{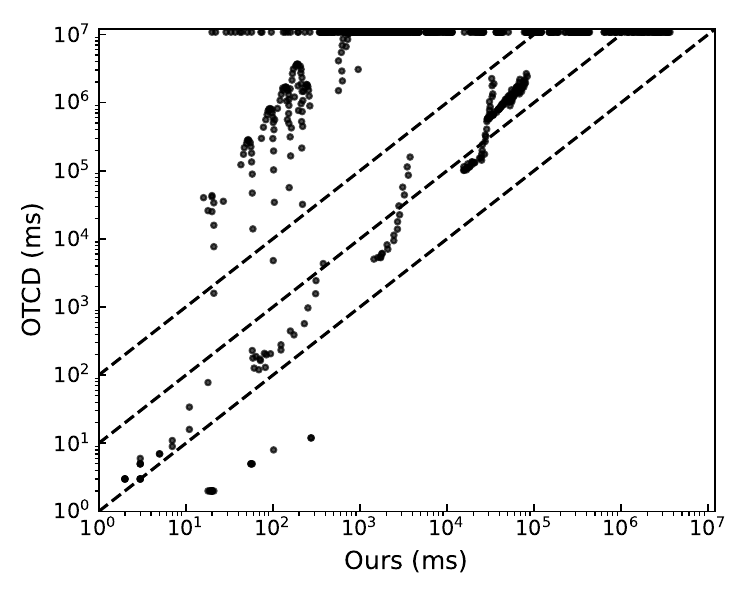}
    \caption{Scatter plot of running times for \texttt{CoreT} and \texttt{OTCD}.}
    \label{fig:scatter-runtime}
\end{figure}

We first present a unified scatter plot (Figure~\ref{fig:scatter-runtime}) that illustrates the running times of both \texttt{OTCD} and our proposed algorithm \texttt{CoreT} across all test instances. This provides a comprehensive performance overview. As shown in the figure, the vast majority of points lie significantly above the diagonal, often by orders of magnitude, demonstrating that \texttt{CoreT} consistently achieves much lower running times compared to \texttt{OTCD}. In many cases, the points corresponding to \texttt{OTCD} appear near or beyond the upper bound of the plot, indicating timeouts (i.e., running time exceeding three hours), while the corresponding \texttt{CoreT} times remain well within seconds or minutes. We remark that while \texttt{OTCD} is marginally faster than \texttt{CoreT} in a small number of specific instances, as indicated by points slightly below the diagonal, \texttt{CoreT} still completes within one second in all such cases.

\smallskip
\noindent \underline{\textbf{Exp-2: Time comparison for 5-core with finer intervals.}}
To further illustrate the effect of interval length on performance, we conduct a fine-grained scalability study using the CM and EE datasets (Figure~\ref{fig:sweep}). These datasets are chosen since DC, WG, and FG have $t_{\max}$ values that are too small for fine slicing, whereas \texttt{OTCD} times out on larger datasets such as SU, MO, AU, and WT. Specifically, we fix \(k=5\) and vary \(\alpha\) from 0 to 1 in increments of 0.01. This results in approximately 600 and 2,000 additional timestamps per step for CM and EE datasets, respectively, thereby providing a detailed view of runtime trends as temporal coverage increases.
From Figure~\ref{fig:sweep}, we can see that both two datasets exhibit similar runtime trends. Across all interval lengths, \texttt{CoreT} consistently outperforms \texttt{OTCD}. As the temporal interval length increases (i.e., with increasing $\alpha$), the runtime of \texttt{OTCD} grows rapidly. For example, \texttt{OTCD} exceeds the 3-hour timeout threshold when the number of timestamps reaches approximately 120K in the EE dataset, while \texttt{CoreT} remains efficient throughout the entire range.

\begin{figure}[t]
    \centering
    \begin{subfigure}
        [t]{0.5\linewidth}
        \centering
        {\includegraphics[width=\textwidth]{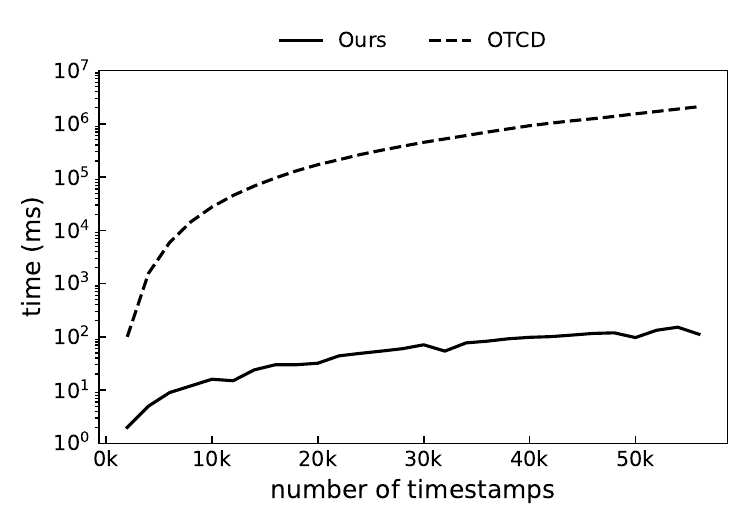}}
        \caption{CM}
        \label{fig:timestamp-sweep-CM}
    \end{subfigure}\hspace{-0.7em}
    \begin{subfigure}
        [t]{0.5\linewidth}
        \centering
        {\includegraphics[width=\textwidth]{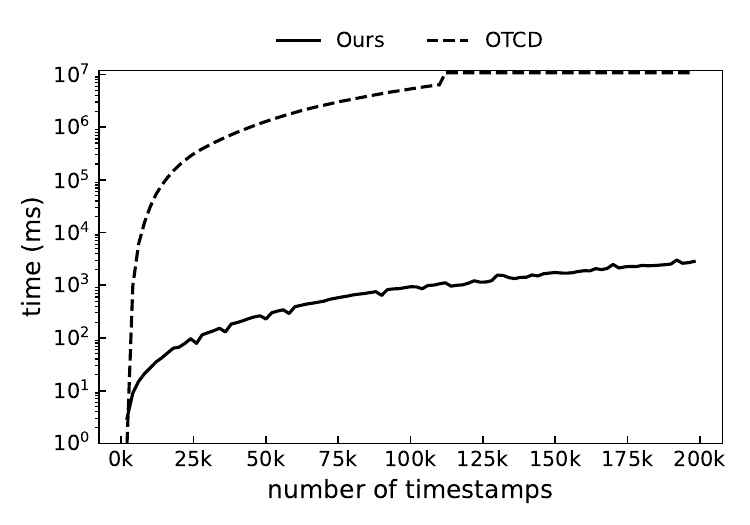}}
        \caption{EE}
        \label{fig:timestamp-sweep-EE}
    \end{subfigure}
    \caption{Runtime comparison over increasing time interval.}
    \label{fig:sweep}
\end{figure}

\begin{figure}
    \centering
    \begin{subfigure}
        [t]{0.5\linewidth}
        \centering
        {\includegraphics[width=\textwidth]{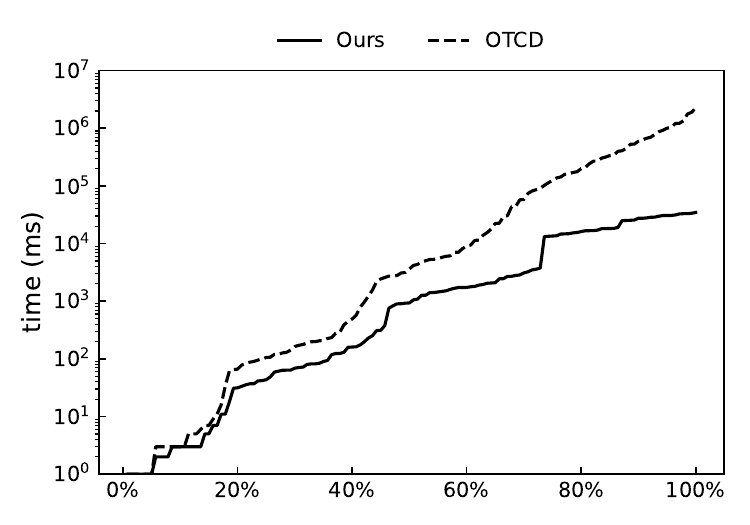}}
        \caption{DC}
        \label{fig:timestamp-sweep-CM}
    \end{subfigure}\hspace{-0.7em}
    \begin{subfigure}
        [t]{0.5\linewidth}
        \centering
        {\includegraphics[width=\textwidth]{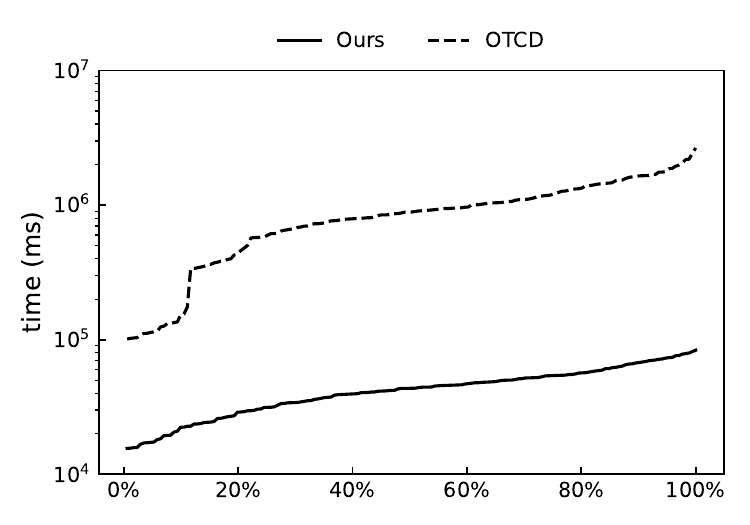}}
        \caption{FG}
        \label{fig:timestamp-sweep-EE}
    \end{subfigure}
    \caption{Cumulative distribution of query response times.}

    \label{fig:cdf}
\end{figure}

\smallskip
\noindent \underline{\textbf{Exp-3: Time comparison with random queries.}}
To rigorously compare the query efficiency of \texttt{CoreT} and \texttt{OTCD} under short query intervals, we conduct a focused empirical study on the DC and FG datasets. These datasets are selected for their relatively small values of $t_{\max}$ (i.e., 77 and 134), which help minimize variance caused by interval length randomness and enable a fairer comparison between the two methods. For each dataset, we randomly generate 500 queries with varying interval lengths and $k$ ranging from 2 to 20, excluding queries without valid $k$-core results.

Figure~\ref{fig:cdf} presents the cumulative distribution functions (CDFs) of query times for both algorithms. The horizontal axis represents query time, and the vertical axis indicates the percentage of queries completed within that time. At lower query times, both algorithms show similar efficiency; however, as query time increases, the curve for \texttt{OTCD} consistently remains below that of \texttt{CoreT}, indicating that \texttt{CoreT} completes a larger proportion of queries in less time overall. Although \texttt{OTCD} occasionally achieves slightly faster responses in especially sparse intervals, such cases are rare and the magnitude of improvement is negligible. Overall, while \texttt{OTCD} can be effective under narrowly constrained conditions, \texttt{CoreT} demonstrates broader reliability and efficiency across diverse temporal configurations.

\smallskip
\noindent \underline{\textbf{Exp-4: Detailed time comparison with the baseline.}}
We show the detailed runtime analyses in Figures~\ref{fig:core-ColledgeMsg} to~\ref{fig:core-wikipedia-growth}. Note that no cores are identified for the DC dataset at $\alpha=0.2$; hence, there are no corresponding results shown in Figure~\ref{fig:dblp-coauthor}. 

We have the following observations. 
\textbf{First}, \texttt{CoreT} demonstrates consistently strong performance across all datasets and parameter settings, exhibiting excellent scalability and runtime stability. In contrast, \texttt{OTCD} suffers significant slowdowns as the interval length (i.e., $\alpha$) increases, and frequently exceeds the three-hour time limit on large or dense graphs. In many scenarios, \texttt{CoreT} achieves speedups of several orders of magnitude. For instance, on the CM dataset with full temporal span ($\alpha=1.0$) in Figure~\ref{fig:core-ColledgeMsg}(e), our algorithm finishes in only 144 ms for $k=2$, whereas \texttt{OTCD} takes 1,054,508 ms, yielding a speedup of more than 7,300 times. As $k$ increases to 20, the runtime of \texttt{CoreT} remains highly stable (144 -- 220 ms), while the runtime of \texttt{OTCD} continues reaches over $3.11 \times 10^6$ ms.

\textbf{Second}, both algorithms tend to run faster as the core threshold $k$ increases in most cases, which is consistent with the fact that larger values of $k$ typically induce smaller subgraphs with fewer edges to process. This trend is clearly observable on the DC and FG datasets (Figures~\ref{fig:dblp-coauthor} and~\ref{fig:flickr-growth}), where the runtime curves consistently decrease with increasing $k$ across all tested values of $\alpha$. 

\textbf{Third}, although \texttt{OTCD} can outperform \texttt{CoreT} slightly in a few cases where the running time is short, the differences are marginal. On smaller datasets or datasets with very few timestamps, such as CM with $\alpha=0.2$ and $k \ge 12$ (Figure~\ref{fig:core-ColledgeMsg}(a)), \texttt{OTCD} is faster by approximately 1 ms. A similar effect is observed on the DC dataset at $\alpha=0.4$ in Figure~\ref{fig:dblp-coauthor}(b), which contains only 30 distinct timestamps. In such cases, \texttt{OTCD} benefits from its simplicity and low overhead. However, the advantage is minimal (within 1 ms), and quickly vanishes as the number of timestamps or graph size increases. 

\textbf{Finally}, the size of the query interval has a strong impact on the performance of \texttt{OTCD}, whereas \texttt{CoreT} maintains consistently high efficiency. As the interval grows, the running time of \texttt{OTCD} increases rapidly, often exceeding the allowable time limit on larger intervals. For instance, on the medium-sized EE dataset in Figure~\ref{fig:email-Eu-core}, \texttt{OTCD} succeeds only when $\alpha \le 0.4$, but fails to complete within the 3-hour limit for larger values of $\alpha$. In contrast, \texttt{CoreT} completes all five intervals on EE in under 5 seconds. This sensitivity becomes even more obvious on large-scale graphs such as AU, MO, and WT, where \texttt{OTCD} fails to complete any instance, whereas \texttt{CoreT} consistently completes all queries in under 120 seconds.

\begin{figure*}
    \centering
    \subcaptionbox{$\alpha=0.2$}{ \includegraphics[width=0.195\textwidth]{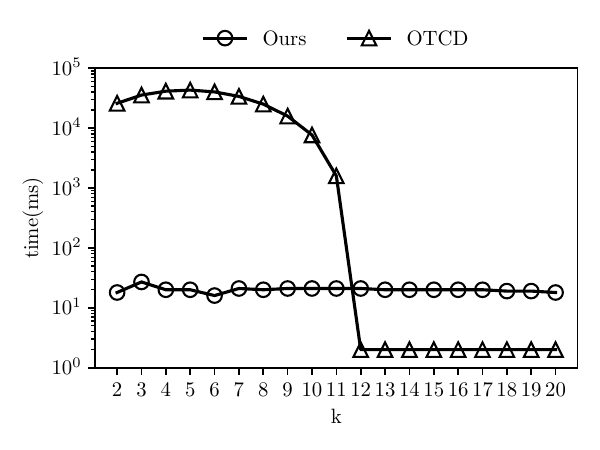} }\hspace{-0.7em}%
    \subcaptionbox{$\alpha=0.4$}{ \includegraphics[width=0.195\textwidth]{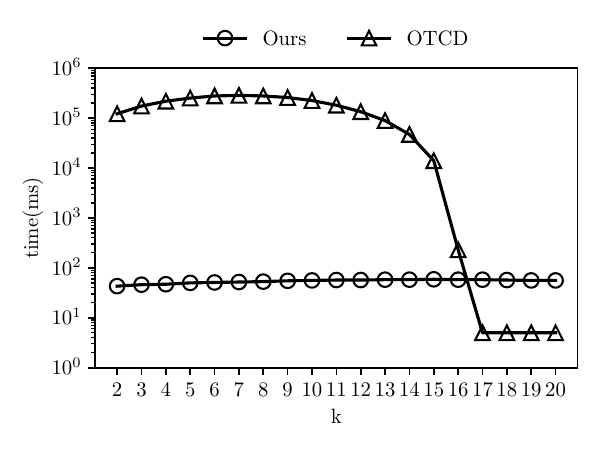} }\hspace{-0.7em}%
    \subcaptionbox{$\alpha=0.6$}{ \includegraphics[width=0.195\textwidth]{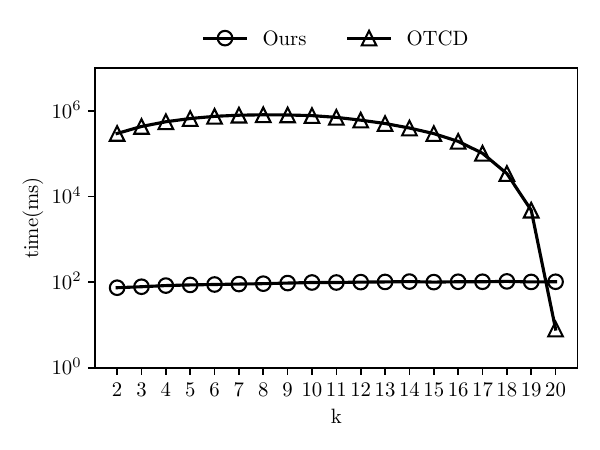} }\hspace{-0.7em}%
    \subcaptionbox{$\alpha=0.8$}{ \includegraphics[width=0.195\textwidth]{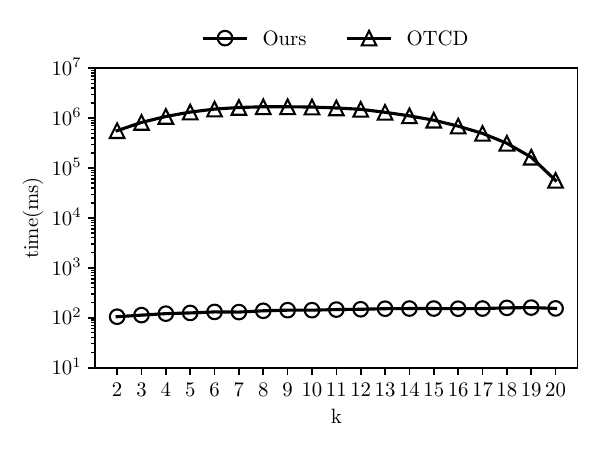} }\hspace{-0.7em}%
    \subcaptionbox{$\alpha=1.0$}{ \includegraphics[width=0.195\textwidth]{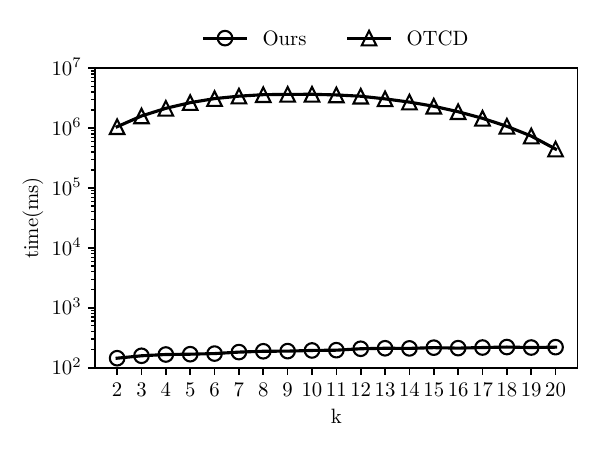} }\hspace{-0.7em}%
    \caption{Runtime performance of \texttt{CoreT} and \texttt{OTCD} on CollegeMsg (CM)}
    \label{fig:core-ColledgeMsg}

\end{figure*}
\begin{figure*}
    \centering
    \subcaptionbox{$\alpha=0.2$}{ \includegraphics[width=0.195\textwidth]{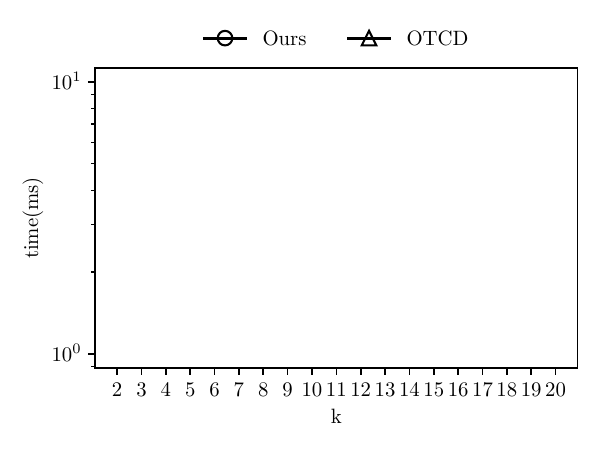} }\hspace{-0.7em}%
    \subcaptionbox{$\alpha=0.4$}{ \includegraphics[width=0.195\textwidth]{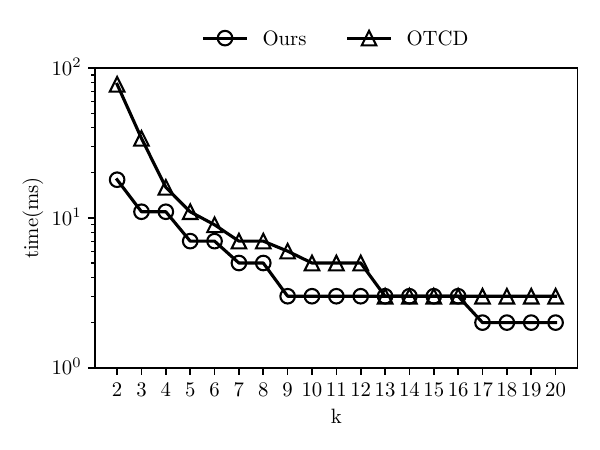} }\hspace{-0.7em}%
    \subcaptionbox{$\alpha=0.6$}{ \includegraphics[width=0.195\textwidth]{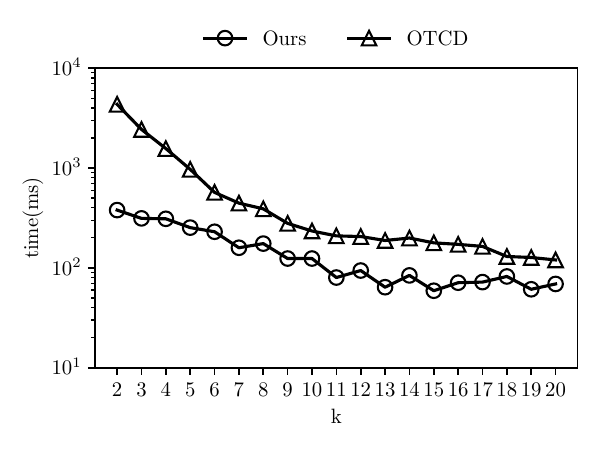} }\hspace{-0.7em}%
    \subcaptionbox{$\alpha=0.8$}{ \includegraphics[width=0.195\textwidth]{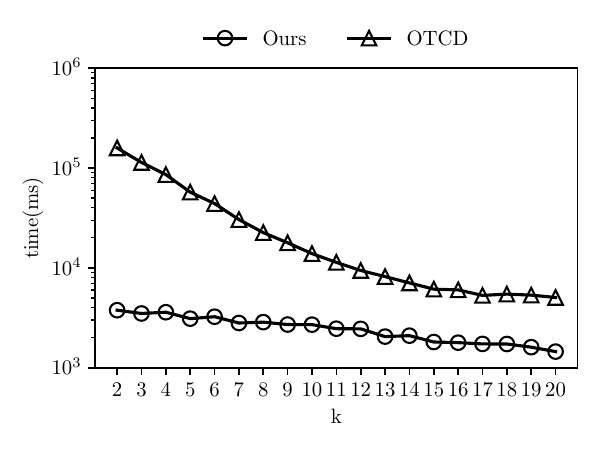} }\hspace{-0.7em}%
    \subcaptionbox{$\alpha=1.0$}{ \includegraphics[width=0.195\textwidth]{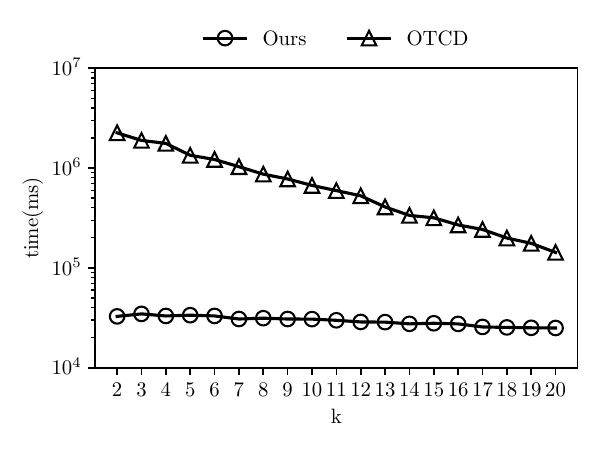} }\hspace{-0.7em}%
    \caption{Runtime performance of \texttt{CoreT} and \texttt{OTCD} on dblp-coauthor (DC)}
    \label{fig:dblp-coauthor}
\end{figure*}
\begin{figure*}
    \centering
    \subcaptionbox{$\alpha=0.2$}{ \includegraphics[width=0.195\textwidth]{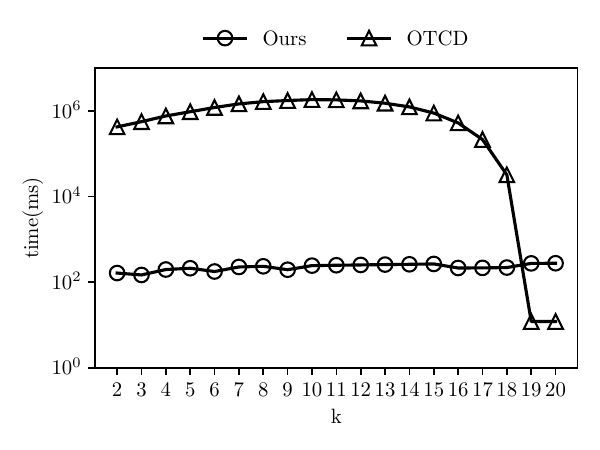} }\hspace{-0.7em}%
    \subcaptionbox{$\alpha=0.4$}{ \includegraphics[width=0.195\textwidth]{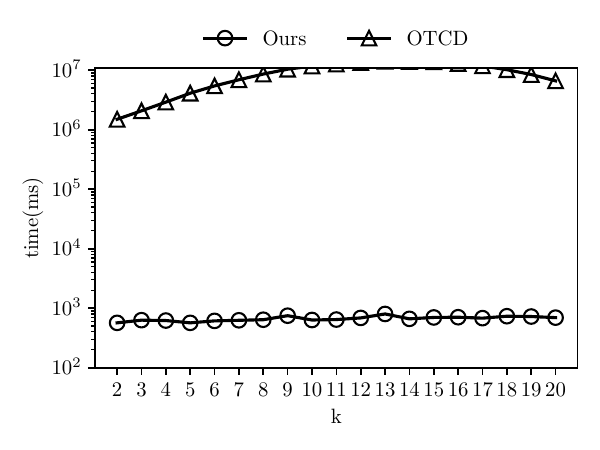} }\hspace{-0.7em}%
    \subcaptionbox{$\alpha=0.6$}{ \includegraphics[width=0.195\textwidth]{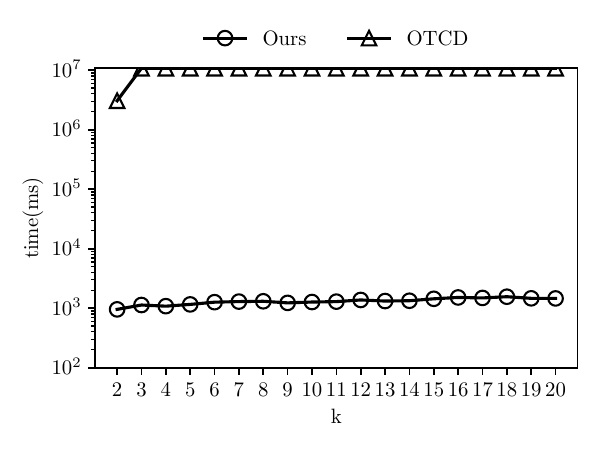} }\hspace{-0.7em}%
    \subcaptionbox{$\alpha=0.8$}{ \includegraphics[width=0.195\textwidth]{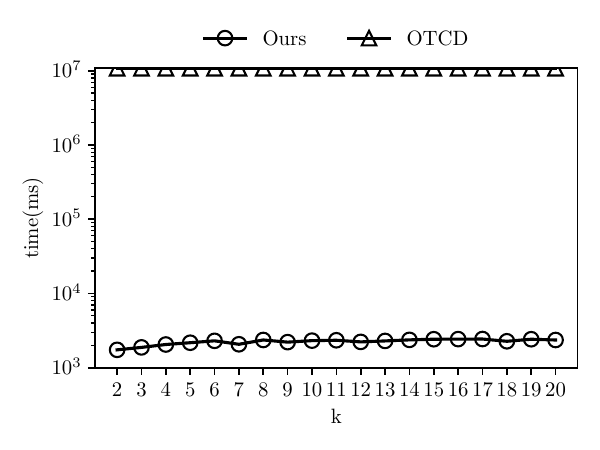} }\hspace{-0.7em}%
    \subcaptionbox{$\alpha=1.0$}{ \includegraphics[width=0.195\textwidth]{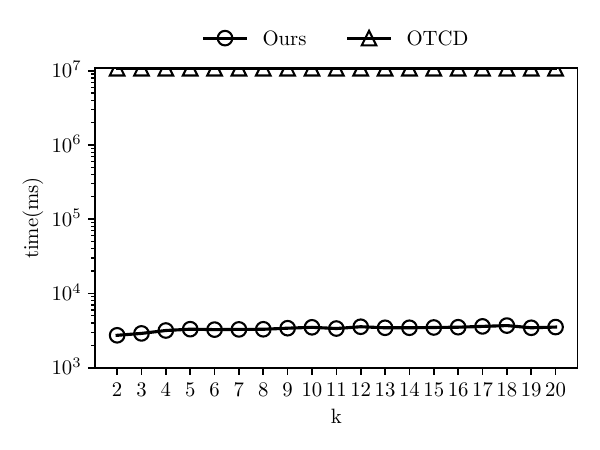} }\hspace{-0.7em}%
    \caption{Runtime performance of \texttt{CoreT} and \texttt{OTCD} on email-Eu-core-temporal (EE)}
    \label{fig:email-Eu-core}
\end{figure*}
\begin{figure*}
    \centering
    \subcaptionbox{$\alpha=0.2$}{ \includegraphics[width=0.195\textwidth]{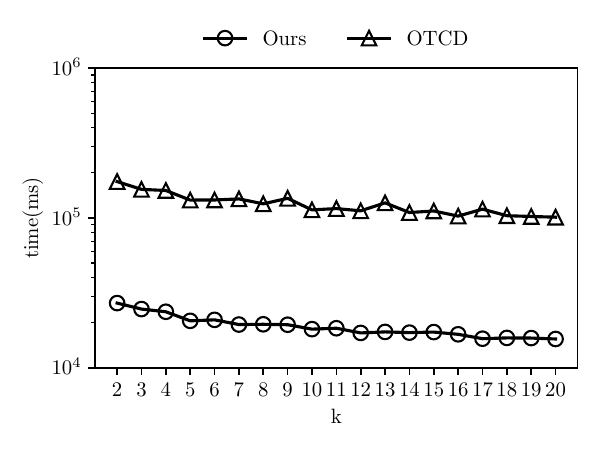} }\hspace{-0.7em}%
    \subcaptionbox{$\alpha=0.4$}{ \includegraphics[width=0.195\textwidth]{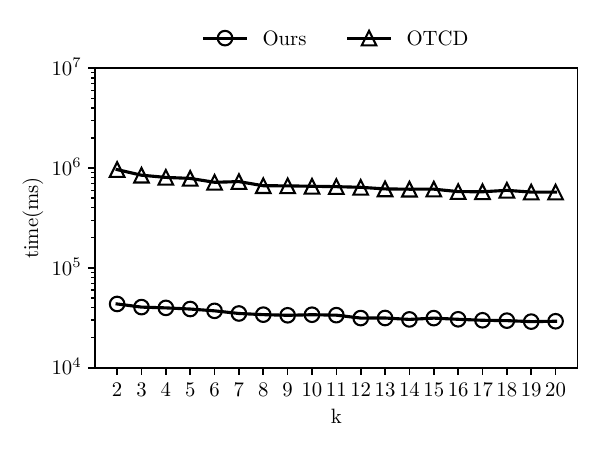} }\hspace{-0.7em}%
    \subcaptionbox{$\alpha=0.6$}{ \includegraphics[width=0.195\textwidth]{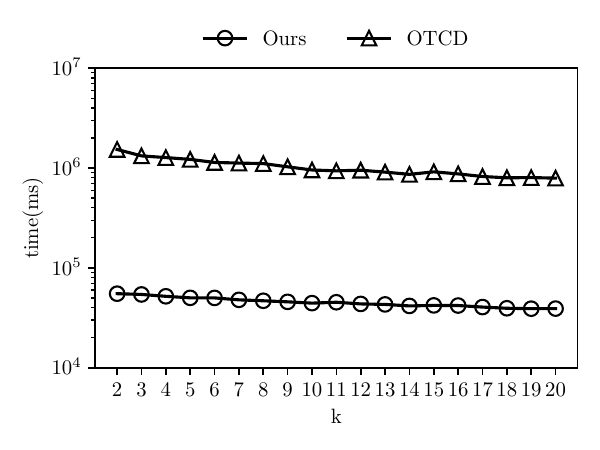} }\hspace{-0.7em}%
    \subcaptionbox{$\alpha=0.8$}{ \includegraphics[width=0.195\textwidth]{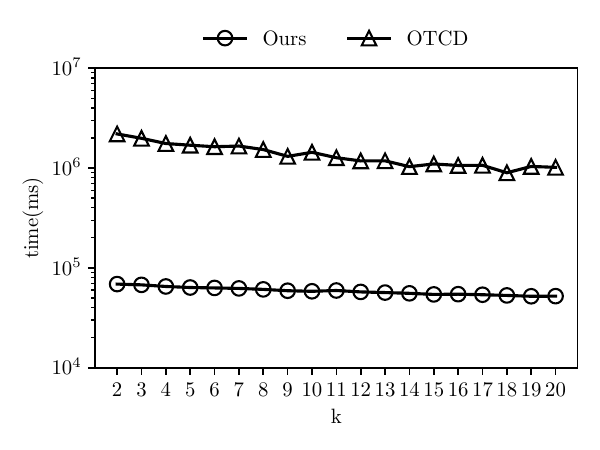} }\hspace{-0.7em}%
    \subcaptionbox{$\alpha=1.0$}{ \includegraphics[width=0.195\textwidth]{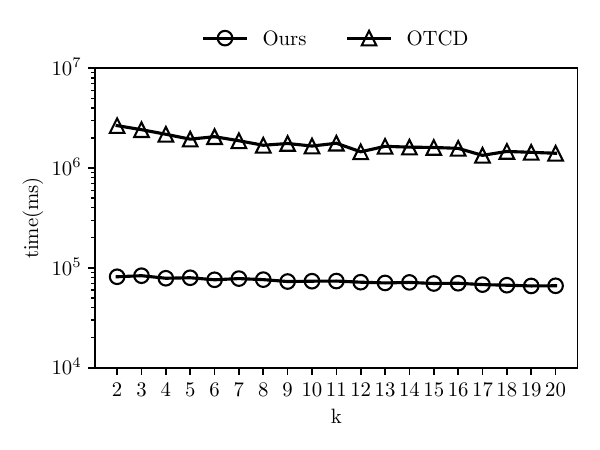} }\hspace{-0.7em}%
    \caption{Runtime performance of \texttt{CoreT} and \texttt{OTCD} on flickr-growth (FG)}
    \label{fig:flickr-growth}
\end{figure*}
\begin{figure*}
    \centering
    \subcaptionbox{$\alpha=0.2$}{ \includegraphics[width=0.195\textwidth]{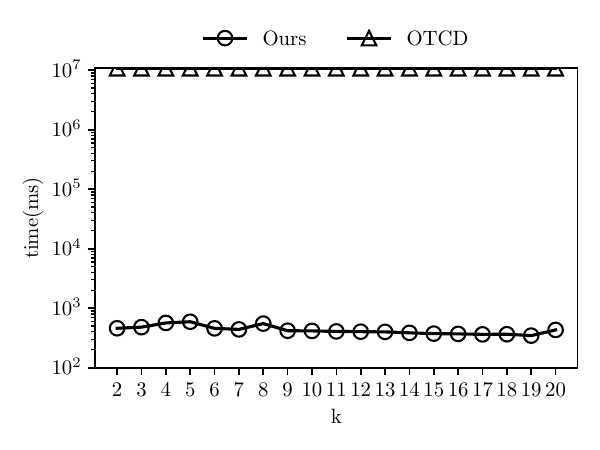} }\hspace{-0.7em}%
    \subcaptionbox{$\alpha=0.4$}{ \includegraphics[width=0.195\textwidth]{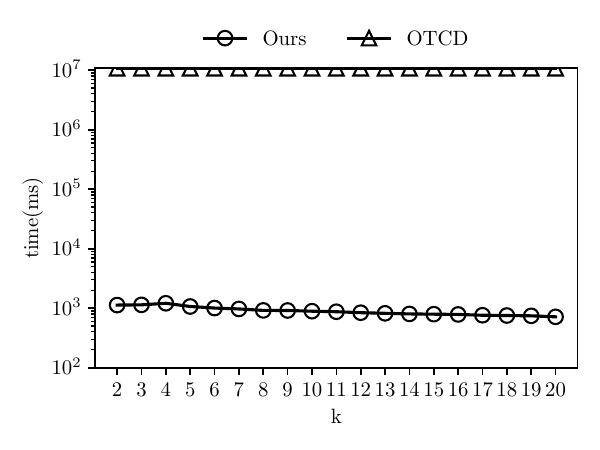} }\hspace{-0.7em}%
    \subcaptionbox{$\alpha=0.6$}{ \includegraphics[width=0.195\textwidth]{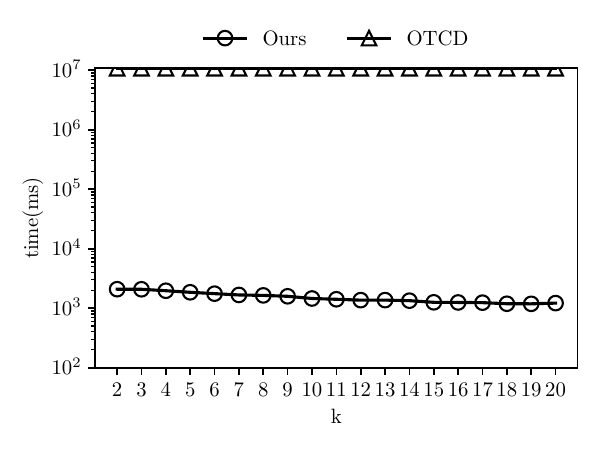} }\hspace{-0.7em}%
    \subcaptionbox{$\alpha=0.8$}{ \includegraphics[width=0.195\textwidth]{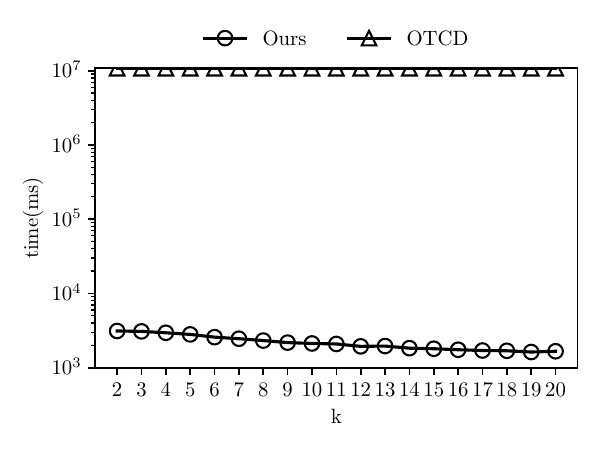} }\hspace{-0.7em}%
    \subcaptionbox{$\alpha=1.0$}{ \includegraphics[width=0.195\textwidth]{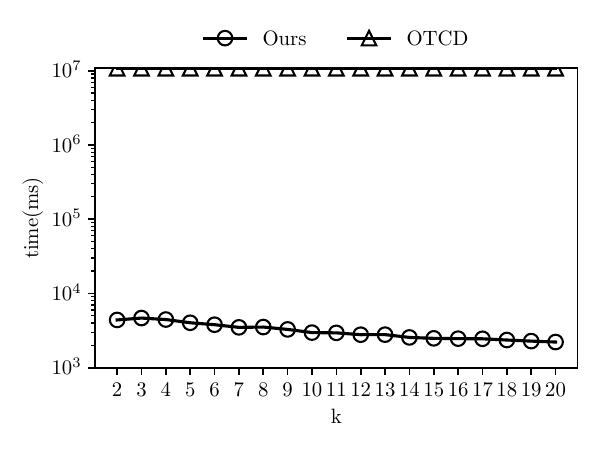} }\hspace{-0.7em}%
    \caption{Runtime performance of \texttt{CoreT} and \texttt{OTCD} on sx-askubuntu (AU)}
\end{figure*}
\begin{figure*}
    \centering
    \subcaptionbox{$\alpha=0.2$}{ \includegraphics[width=0.195\textwidth]{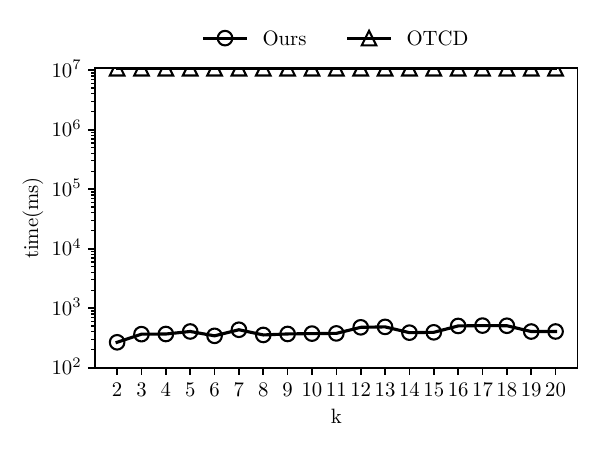} }\hspace{-0.7em}%
    \subcaptionbox{$\alpha=0.4$}{ \includegraphics[width=0.195\textwidth]{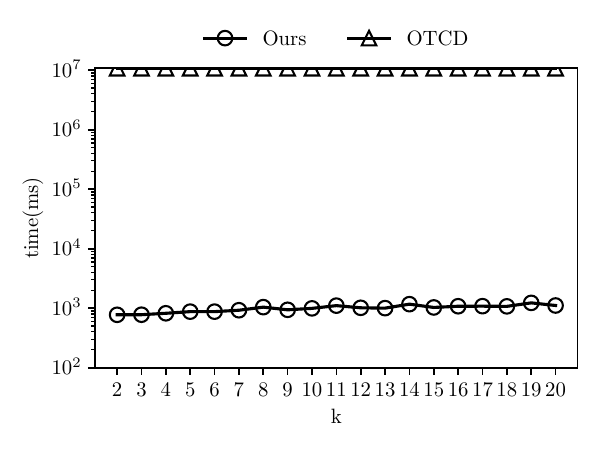} }\hspace{-0.7em}%
    \subcaptionbox{$\alpha=0.6$}{ \includegraphics[width=0.195\textwidth]{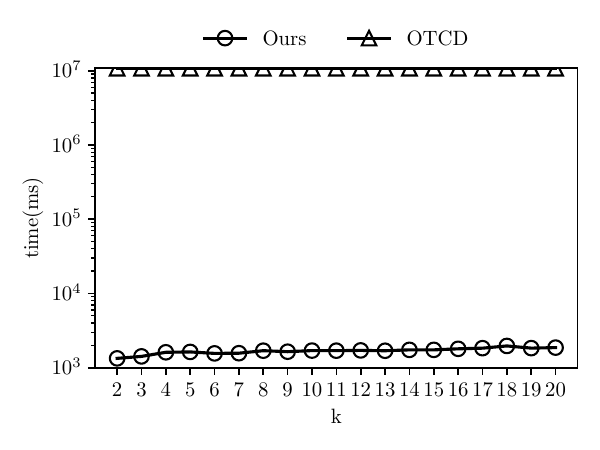} }\hspace{-0.7em}%
    \subcaptionbox{$\alpha=0.8$}{ \includegraphics[width=0.195\textwidth]{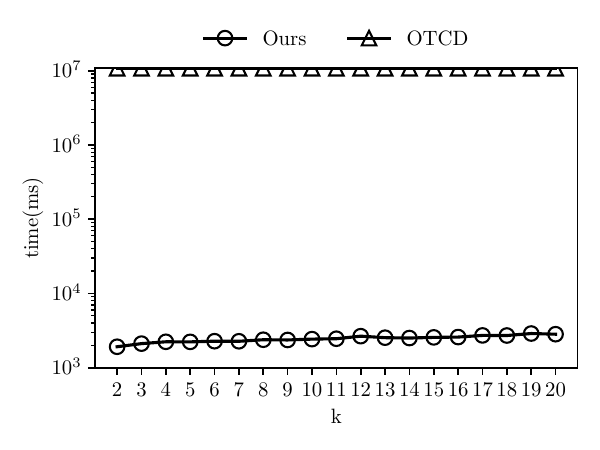} }\hspace{-0.7em}%
    \subcaptionbox{$\alpha=1.0$}{ \includegraphics[width=0.195\textwidth]{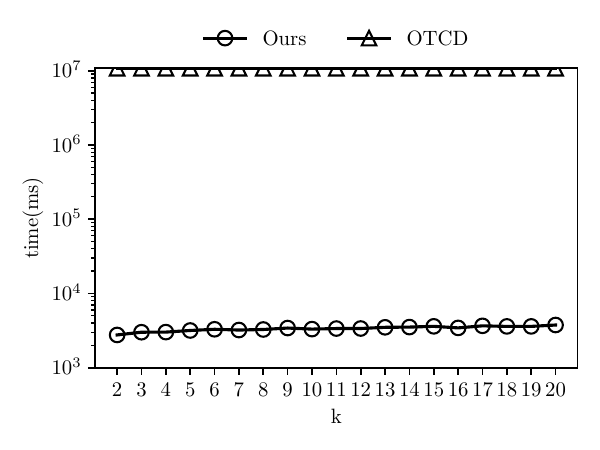} }\hspace{-0.7em}%
    \vspace{-0.1in}
    \caption{Runtime performance of \texttt{CoreT} and \texttt{OTCD} on sx-mathoverflow (MO)}
    \vspace{-0.1in}
\end{figure*}
\begin{figure*}
    \centering
    \subcaptionbox{$\alpha=0.2$}{ \includegraphics[width=0.195\textwidth]{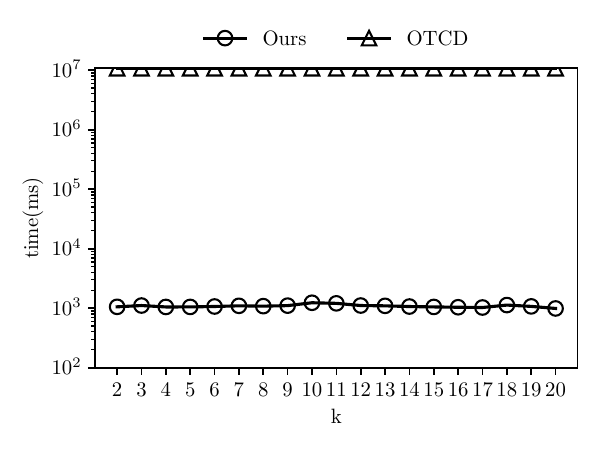} }\hspace{-0.7em}%
    \subcaptionbox{$\alpha=0.4$}{ \includegraphics[width=0.195\textwidth]{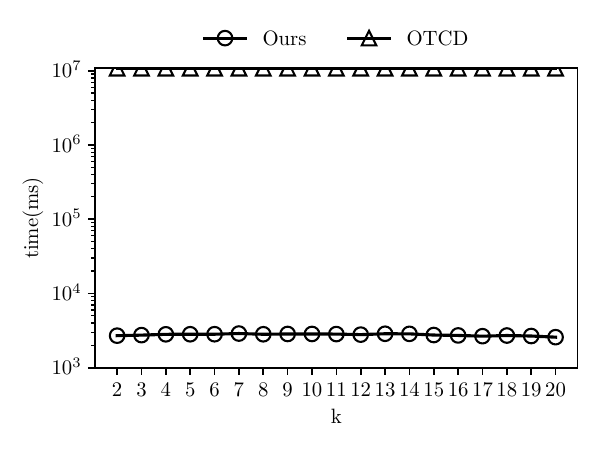} }\hspace{-0.7em}%
    \subcaptionbox{$\alpha=0.6$}{ \includegraphics[width=0.195\textwidth]{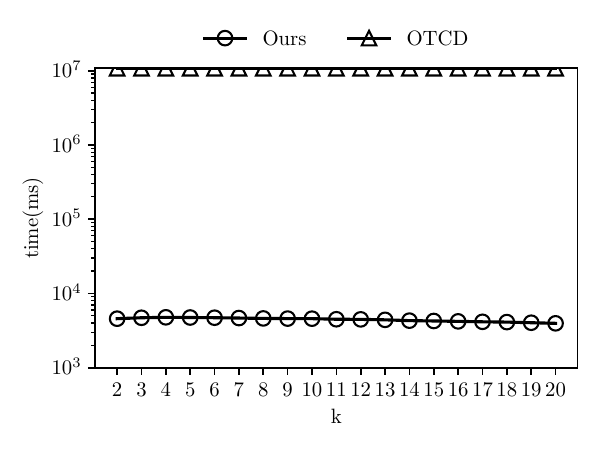} }\hspace{-0.7em}%
    \subcaptionbox{$\alpha=0.8$}{ \includegraphics[width=0.195\textwidth]{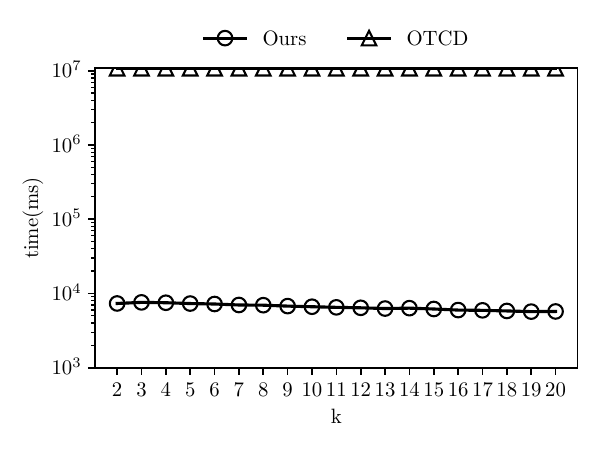} }\hspace{-0.7em}%
    \subcaptionbox{$\alpha=1.0$}{ \includegraphics[width=0.195\textwidth]{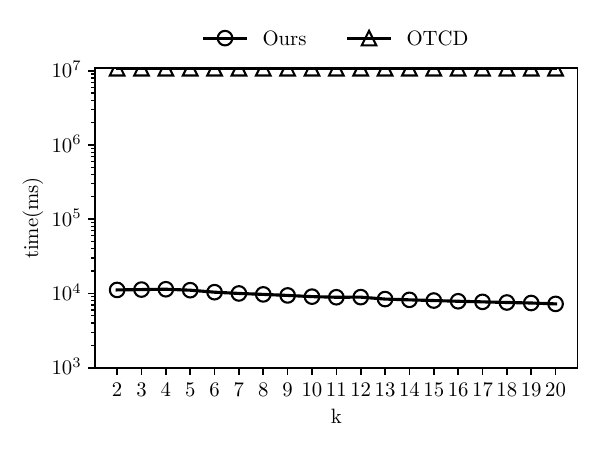} }\hspace{-0.7em}%
    \vspace{-0.1in}
    \caption{Runtime performance of \texttt{CoreT} and \texttt{OTCD} on sx-superuser (SU)}
    \vspace{-0.1in}
\end{figure*}
\begin{figure*}
    \centering
    \subcaptionbox{$\alpha=0.2$}{ \includegraphics[width=0.195\textwidth]{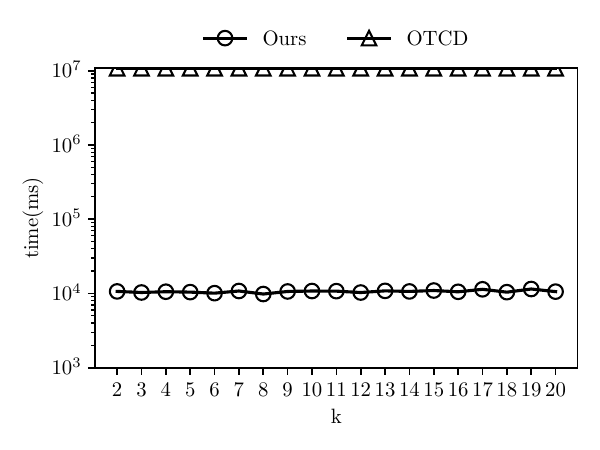} }\hspace{-0.7em}%
    \subcaptionbox{$\alpha=0.4$}{ \includegraphics[width=0.195\textwidth]{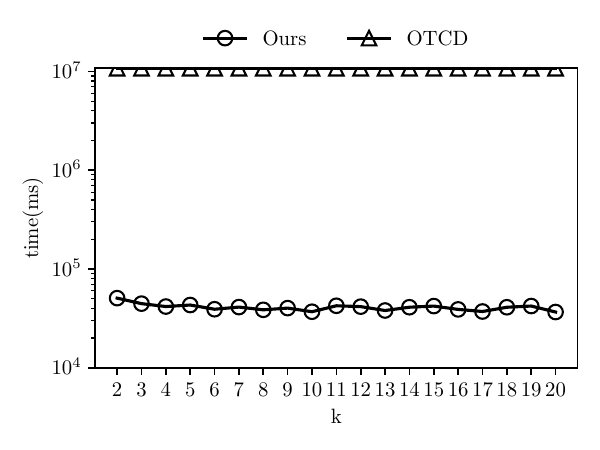} }\hspace{-0.7em}%
    \subcaptionbox{$\alpha=0.6$}{ \includegraphics[width=0.195\textwidth]{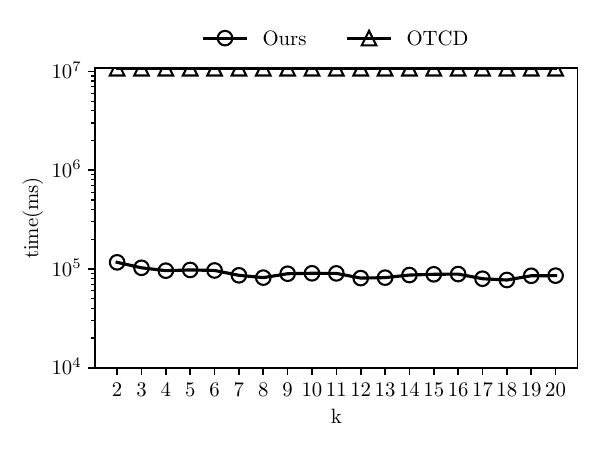} }\hspace{-0.7em}%
    \subcaptionbox{$\alpha=0.8$}{ \includegraphics[width=0.195\textwidth]{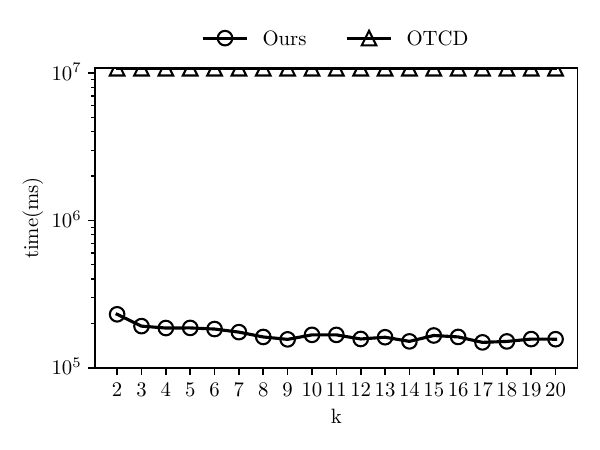} }\hspace{-0.7em}%
    \subcaptionbox{$\alpha=1.0$}{ \includegraphics[width=0.195\textwidth]{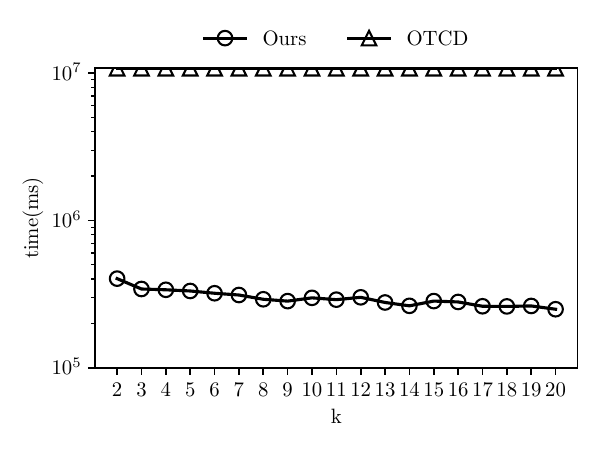} }\hspace{-0.7em}%
    \vspace{-0.1in}
    \caption{Runtime performance of \texttt{CoreT} and \texttt{OTCD} on wiki-talk-temporal (WT)}
    \vspace{-0.1in}
\end{figure*}
\begin{figure*}
    \centering
    \label{fig:core2}
    \subcaptionbox{$\alpha=0.2$}{ \includegraphics[width=0.195\textwidth]{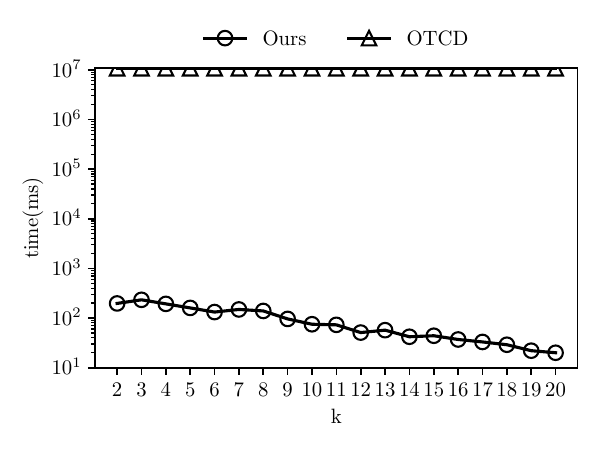} }\hspace{-0.7em}%
    \subcaptionbox{$\alpha=0.4$}{ \includegraphics[width=0.195\textwidth]{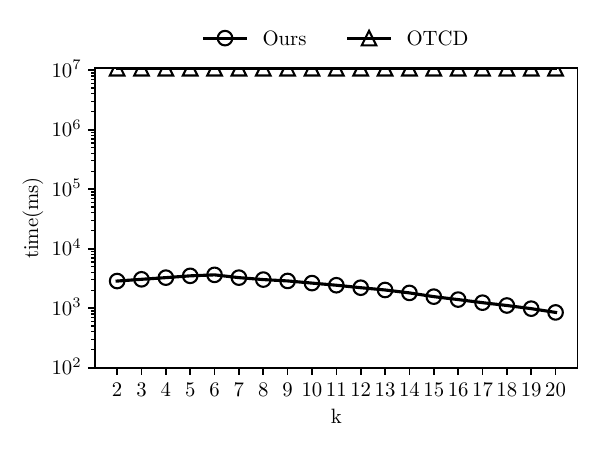} }\hspace{-0.7em}%
    \subcaptionbox{$\alpha=0.6$}{ \includegraphics[width=0.195\textwidth]{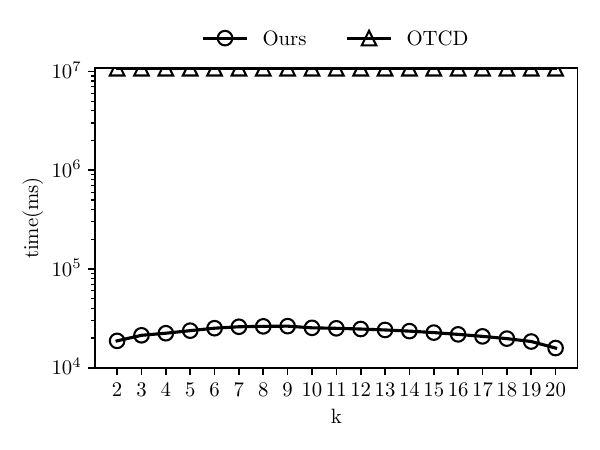} }\hspace{-0.7em}%
    \subcaptionbox{$\alpha=0.8$}{ \includegraphics[width=0.195\textwidth]{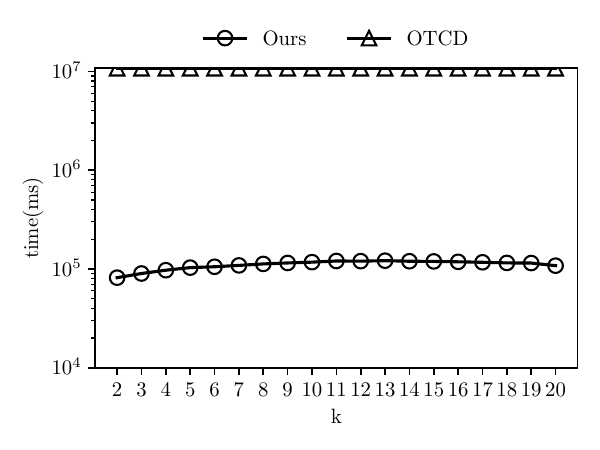} }\hspace{-0.7em}%
    \subcaptionbox{$\alpha=1.0$}{ \includegraphics[width=0.195\textwidth]{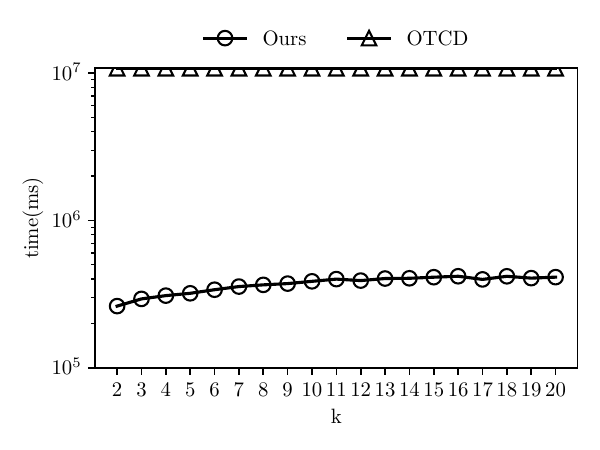} }\hspace{-0.7em}%
    \vspace{-0.1in}
    \caption{Runtime performance of \texttt{CoreT} and \texttt{OTCD} on wikipedia-growth (WG)}
    \label{fig:core-wikipedia-growth}
\end{figure*}

\section{Related Work}
\label{sec:related}

Temporal $k$-core queries and decompositions have attracted significant attention in recent years. Comprehensive surveys of temporal graphs are provided in~\cite{holme2015modern,wang2019time}, while detailed reviews on core decomposition can be found in~\cite{kong2019k,malliaros2020core}.

\noindent
\underline{\textbf{$k$-core variants in temporal graphs.}} The work most closely related to ours is presented in~\cite{yang2023scalable}, which addresses the same problem. Building upon this, the authors subsequently developed an index structure in~\cite{yang2024evolution} to support temporal $k$-core component search for any given vertex and time interval.

Beyond the temporal $k$-core, several other variants of $k$-core models were proposed for temporal graphs. For instance, Yu et al.~\cite{yu2021querying} introduced the concept of historical $k$-core and developed the PHC-index to efficiently identify such structures. Building on this idea, Tan et al.~\cite{tan2024efficient} generalized the PHC-index to support historical $k$-truss queries. Subsequently, Wang et al.~\cite{wangmore} proposed optimizations to the PHC-index for scenarios where the resulting historical $k$-core is small. In addition, Wu et al.~\cite{wu2015core} introduced the $(k,h)$-core, which requires at least $h$ parallel edges between neighbors, while Bai et al.~\cite{BAI2020324} investigated its maintenance in massive temporal graphs. Galimberti et al.~\cite{galimberti2018mining} defined span-cores, which require edge persistence throughout a continuous interval, whereas Li et al.~\cite{li2018persistent} introduced persistent $k$-cores that have to remain valid for a specified time span. By abstracting these models together with the temporal $k$-core, Conte et al.~\cite{conte2024k} proposed the unified $(k,h,\Delta)$-core, demonstrating its ability to reveal meaningful structures in temporal graphs. Furthermore, Oettershagen et al.~\cite{oettershagen2025edge} introduced the $(k,\Delta)$-core and $(k,\Delta)$-truss, and present a general decomposition framework for their identification. Temporal communities have also been explored through alternative definitions. For example, Qin et al.~\cite{qin2022mining} defined the $(\ell,\delta)$-maximal dense core, which requires a minimum density sustained over a time interval of length $\ell$, while Huang et al.~\cite{hung2021maximum} introduced the $(L,k)$-lasting core, a vertex set forming a $k$-core that persists for $L$ time units. Other studies have examined time-aware properties such as continuity \cite{li2021efficient} and reliability \cite{tang2022reliable}. In a different direction, Momin et al.~\cite{momin2023kwiq} introduced the notion of core-invariant nodes, which are nodes whose core numbers remain above a given threshold throughout a specified time window.

\noindent
\underline{\textbf{Other temporal cohesive subgraph models.}} Beyond $k$-core models, numerous cohesive subgraph models have been extended to temporal settings. Among these, periodic cliques were first studied by~\cite{li2018persistent}; subsequently, this notion was generalized into quasi-periodic cliques by~\cite{qin2020periodic}. With regard to temporal cliques, Himmel et al.~\cite{himmel2016enumerating} enumerated all maximal $\Delta$-cliques, defined as vertex subsets where every pair is connected within some time interval of length $\Delta$. Later, Brunelli et al.~\cite{brunelli2025output} present an output-sensitive algorithm for enumerating maximal cliques in temporal graphs under the single-interval model, achieving polynomial delay per output and incremental polynomial time with additional space. 
In another line of work, Xie et al.~\cite{xie2023querying} explored connected component queries in both temporal directed and undirected graphs. 
Finally, Chen et al.~\cite{chen2024querying} considered the maintenance of structural diversity among vertices in temporal graphs, which addresses a different aspect of temporal graph analysis.
\section{Conclusion}
In this work, we revisited the temporal $k$-core query problem and introduced \texttt{CoreT}. Specifically, \texttt{CoreT} is a novel core-time-based algorithm whose time complexity is linear in both the number of temporal edges within the query interval and the duration of the interval. Experiments on real-world datasets show that \texttt{CoreT} achieves up to four orders of magnitude speedup over the state-of-the-art while maintaining robust performance across diverse graph sizes and temporal settings. In future work, we plan to extend our approach to enumerate higher-order cohesive structures in temporal graphs, such as temporal $k$-trusses.

\cleardoublepage
\bibliographystyle{IEEEtranS}
\bibliography{ref}

\end{document}